\documentclass[aps,pra,preprint,tightenlines,notitlepage,twocolumn,superscriptaddress,nofootinbib,10pt]{revtex4-2}

\usepackage[percent]{overpic}
\usepackage[english]{babel}
\usepackage{amsmath,amssymb,amsthm,amscd}
\usepackage{dsfont,hyperref,commath,mathtools,thmtools}
\usepackage{mathrsfs}
\usepackage{tcolorbox}
\usepackage{tikz}
\usetikzlibrary{arrows,positioning,shapes,shadows}

\definecolor{bluegray}{RGB}{230,230,255}
\definecolor{paleyellow}{RGB}{255,255,204}
\tikzstyle{block}=[draw=black, rectangle, align=center, fill=bluegray, drop shadow]
\tikzstyle{rblock}=[draw=black, shape=rectangle,rounded corners=1em,align=center,fill=paleyellow, drop shadow]

\DeclareMathOperator*{\argmin}{argmin}
\DeclareMathOperator{\Tr}{Tr}
\DeclareMathOperator{\aff}{aff}
\DeclareMathOperator{\spn}{span}
\DeclareMathOperator{\conv}{conv}
\DeclareMathOperator{\rng}{rng}
\DeclareMathOperator{\supp}{supp}
\DeclareMathOperator{\mvee}{\mathtt{MVEE}}
\DeclareMathOperator{\vol}{\mathcal{V}}
\DeclareMathOperator{\ddi}{\mathtt{ddi}}

\newcommand{\mc}[1]{\mathcal{#1}}
\newcommand{\bb}[1]{\mathbb{#1}}
\newcommand{\N}{\bb{N}}
\newcommand{\R}{\bb{R}}
\newcommand{\M}[2]{\mathfrak{M}_{{#1}  \times {#2}}}

\newtheorem{thm}{Theorem}
\newtheorem{lmm}{Lemma}
\newtheorem{cor}{Corollary}
\newtheorem{prop}{Proposition}
\theoremstyle{definition}
\newtheorem{dfn}{Definition}
\newtheorem*{example*}{Example}

\begin{document}

\title{Data-Driven     Inference,    Reconstruction,     and
  Observational Completeness of Quantum Devices}

\preprint{YITP-20-151}

\date{\today}

\author{Michele \surname{Dall'Arno}}

\email{dallarno.michele@yukawa.kyoto-u.ac.jp}

\affiliation{Yukawa Institute for Theoretical Physics, Kyoto
  University,   Kitashirakawa   Oiwakecho,  Sakyoku,   Kyoto
  606-8502, Japan}

\affiliation{Faculty  of Education  and Integrated  Arts and
  Sciences,    Waseda    University,   1-6-1    Nishiwaseda,
  Shinjuku-ku, Tokyo 169-8050, Japan}

\author{Francesco \surname{Buscemi}}

\email{buscemi@i.nagoya-u.ac.jp}

\affiliation{Graduate  School  of  Informatics,  Nagoya
  University, Chikusa-ku, 464-8601 Nagoya, Japan}

\author{Alessandro \surname{Bisio}}

\affiliation{Quit    group,    Dipartimento    di    Fisica,
  Universit\`a  degli studi  di  Pavia, via  Bassi 6,  27100
  Pavia, Italy}

\affiliation{Istituto Nazionale  di Fisica  Nucleare, Gruppo
  IV, via Bassi 6, 27100 Pavia, Italy}

\author{Alessandro \surname{Tosini}}

\affiliation{Quit    group,    Dipartimento    di    Fisica,
  Universit\`a  degli studi  di  Pavia, via  Bassi 6,  27100
  Pavia, Italy}

\affiliation{Istituto Nazionale  di Fisica  Nucleare, Gruppo
  IV, via Bassi 6, 27100 Pavia, Italy}

\begin{abstract}
  The range of  a quantum measurement is the  set of outcome
  probability distributions that can  be produced by varying
  the input  state. We introduce data-driven  inference as a
  protocol  that, given  a  set of  experimental  data as  a
  collection  of outcome  distributions, infers  the quantum
  measurement which is, i) consistent  with the data, in the
  sense  that  its  range  contains  all  the  distributions
  observed, and,  ii) maximally  noncommittal, in  the sense
  that  its range  is  of  minimum volume  in  the space  of
  outcome distributions.  We show that data-driven inference
  is able to return a unique measurement for any data set if
  and only if the inference adopts a (hyper)-spherical state
  space (for example, the classical or the quantum bit).
  
  In  analogy  to  informational  completeness  for  quantum
  tomography,  we define  observational completeness  as the
  property  of any  set of  states that,  when fed  into any
  given measurement, produces a set of outcome distributions
  allowing for the correct reconstruction of the measurement
  via  data-driven inference.   We  show that  observational
  completeness  is  strictly   stronger  than  informational
  completeness, in  the sense  that not  all informationally
  complete   sets   are   also   observationally   complete.
  Moreover,   we    show   that    for   systems    with   a
  (hyper)-spherical  state space,  the only  observationally
  complete simplex is the regular one, namely, the symmetric
  informationally complete set.
\end{abstract}

\maketitle

\section{Introduction}

In  quantum theory,  as a  consequence of  the Born  rule, a
measurement can always be seen  as a linear mapping from the
set  of states  (i.e., density  operators) into  the set  of
probability     distributions    over     the    measurement
outcomes\footnote{Since the Born rule  is, in fact, bilinear
  in the state-measurement pair,  also the opposite is true,
  namely,  that  any state  induces  a  linear mapping  from
  measurements  into  probability distributions.   For  this
  reason,  in the  appendix the  formalism is  developed for
  both states  and measurements.   However, for the  sake of
  concreteness,  the  narrative  in  the  Main  Text  mostly
  follows  the task  of measurement  inference.}.  In  fact,
some   axiomatic    approaches   \textit{identify}   quantum
measurements with the set of  such mappings: in such a case,
the resulting distribution, i.e., the  image of the state of
the system undergoing the  measurement, receives the natural
operational   interpretation   of  distribution   over   the
measurement outcomes~\cite{OZAWA198011}.

When thinking of measurements  as linear mappings, the image
of   the  set   of   \textit{all}  states   under  a   given
measurement---also known as  the measurement's range---turns
out to  be a very  important mathematical object  in quantum
measurement   theory.   For   example,  given   two  quantum
measurements, the  range of  one includes  the range  of the
other, if and only if the  former can simulate the latter by
means        of         a        suitable        statistical
transformation~\cite{doi:10.1063/1.2008996, 10.2307/2236332,
  Buscemi2012}, independently  of the state  being measured.
Quantum measurements,  hence, can  be compared  by comparing
the   corresponding  ranges,   thus   establishing  a   deep
connection between quantum measurement theory and the theory
of           majorization          and           statistical
comparison~\cite{marshall1979inequalities,  torgersen_1991},
with    ramified   consequences    in   both    theory   and
applications~\cite{cohen1998comparisons,          LIEB19991,
  PhysRevA.93.012101,                Gour:2018aa,DAHL199953,
  shmaya2005comparison,                      SHANNON1958390,
  korner1977comparison,                liese2007statistical,
  csiszar2011information}.

In  this   paper  we  exploit  the   correspondence  between
measurements and their ranges to  propose a method to derive
an  inference about  an unknown  quantum measurement,  based
solely on  the outcome  distributions observed,  without any
knowledge  about the  exact states  that gave  rise to  such
distributions.  As we observe in what follows, such a method
can be naturally divided into  two parts. In the first part,
one  defines  an  inference  rule, which  formulates  in  an
abstract way the rules that  we choose to use when reasoning
in the presence of  incomplete information.  For the problem
at  hand,  such rules  accept  as  input  a set  of  outcome
distributions  and  return  as   output  a  set  of  quantum
measurements.  For  this reason, we name  our inference rule
``data-driven  inference  (DDI) of  quantum  measurements.''
The measurements  inferred via  DDI are consistent  with the
input data and are maximally noncommittal, in the sense that
their  ranges contain  the  input data  and  are of  minimum
volume in the space of  outcome distributions.  DDI need not
aim to infer the ``true'' quantum measurement, as there need
not be any such ``entity'' at this stage. As an application,
we implement an algorithm for  the machine learning of qubit
measurements based on data-driven  inference, and we test it
on data generated by the IBM Q Experience quantum computer.

In the  second part,  one needs  to show  that it  is indeed
possible to construct a real experiment so that DDI leads to
the correct assignment for the unknown measurement. The goal
here  is   reminiscent  of  that  of   conventional  quantum
measurement                                       tomography
(\cite{PhysRevLett.83.3573,PhysRevLett.86.4195,PhysRevLett.93.250407,fiuravsek2001maximum,PhysRevLett.102.010404,
  lundeen2009tomography,feito2009measuring}),   namely,  the
reconstruction of an unknown measurement from the statistics
collected in  a sequence  of experimental  trials.  However,
while the data analysis  performed in measurement tomography
requires the knowledge of the  states that were fed into the
unknown  measurement, DDI  reconstruction only  requires the
analysis   of   the    bare   outcome   distributions:   the
state-preparator  could,  for   example,  emit  a  different
unknown state at each repetition  of the experiment, and DDI
reconstruction  would  still be  applicable\footnote{On  the
  contrary,  we \textit{need}  to  assume  that the  unknown
  measurement to  be reconstructed  remains the  same during
  the   entire   experiment---otherwise   the   problem   of
  reconstruction would not even be well-defined.}.

In  what  follows,  we  expound the  theory  of  data-driven
inference and reconstruction for finite dimensional systems.
As this is based  on the correspondence between measurements
and  their  ranges,  three   main  problems  arise  and  are
addressed.

The first problem is to seek for a general method to infer a
range given  a set of  outcome distributions. As  a possible
solution, in  what follows, we propose  that the measurement
range to be  inferred, in the face of a  set of experimental
data, should be the  \textit{smallest one containing all the
  observed data}.  Recalling that the range of a measurement
is  directly  related with  the  ability  to simulate  other
measurements~\cite{doi:10.1063/1.2008996}, our  principle is
equivalent  to  say that  the  measurements  to be  inferred
should be  the weakest  possible, compatibly with  the data.
Our  inference  rule  hence   encapsulates  a  principle  of
``self-consistent minimality'' that we believe constitutes a
natural  way  to  reason   in  the  presence  of  incomplete
information.   We show  that the  \textit{only} systems  for
which DDI  always leads  to a  unique range  for any  set of
data,   among   all   generalized   probabilistic   theories
\cite{PhysRevA.75.032110,                popescu1994quantum,
  PhysRevA.75.032304,   PhysRevA.71.022101,  Popescu:2014aa,
  d2010testing,    dariano_chiribella_perinotti_2017},   are
those  with  (hyper)-spherical  state  space,  such  as  the
classical and the  quantum bit.  This can  be interpreted as
an ``epistemic reconstruction'' of such systems, regarded as
epistemic  hypotheses  onto  which to  base  our  reasoning,
rather   than   actual    entities   to   be   operationally
characterized    \cite{hardy2001quantum,   dakic2009quantum,
  masanes2011derivation, chiribella2011informational}.

The second problem consists of understanding to which extent
the correspondence  between a measurement and  its range can
be inverted,  that is, to  what extent a measurement  can be
characterized if only its range  is given.  In this respect,
in   what  follows,   we   show   that  the   correspondence
measurement-range is  invertible, but only up  to the action
of a symmetry transformation leaving  the state space of the
system  invariant.  This  is something  to be  expected when
directly working in the  space of outcome distributions, and
we consider this to be  a feature, rather than a limitation,
of DDI.

The third  problem is to understand  how an experimentalist,
in  complete  control  of   their  laboratory,  can  produce
experimental data, which are rich enough to reconstruct, via
DDI, the  ``correct'' range  of a  measurement. That  is, we
want to understand whether, in  order to recover the correct
range by DDI, an infinite set of states needs to be prepared
and  sent through  the measurement  apparatus, or  whether a
finite set  of states,  and possibly the  same ones  for any
measurements, suffice  (we remark  that, in contrast  to the
data analysis  of quantum  measurement tomography,  DDI does
not require the knowledge of  such states).  This problem is
analogous to the problem  in quantum tomography to construct
a  standard  apparatus  that  work  whatever  it  is  to  be
reconstructed.   As the  problem  in  quantum tomography  is
solved by informationally  complete apparatus, the analogous
problem  in DDI  reconstruction is  solved by  what we  call
\textit{observationally  complete}   (OC)  apparatus.   More
precisely, OC sets  of states are sets  whose image contains
the  same statistical  information as  the entire  range. Of
course,  since DDI  does not  rely on  the knowledge  of the
states, it  is not possible  for DDI to certify  whether the
states fed into the measurement were OC.

We show  that the property of  observational completeness is
strictly  stronger  than  informational  completeness,  thus
constituting a new  ``Bureau of Standards'' in  terms of DDI
reconstruction.  To this aim we  show that, for systems with
(hyper)-spherical  state space  such  as  the classical  and
quantum bits,  the only observationally complete  simplex is
the  regular   simplex,  that  is,   the  \textit{symmetric}
informationally                complete                (SIC)
one~\cite{doi:10.1063/1.1737053,        zauner2011grundzuge,
  fuchs2017sic,RevModPhys.85.1693,        mermin2014physics,
  doi:10.1119/1.4874855}.     Data-driven   inference    and
reconstruction, hence,  naturally lead to the  notion of SIC
apparatus \textit{  by looking  only at  the set  of outcome
  distributions}, thus providing  a completely new viewpoint
on the  discussion about  SIC apparatus and  their ``natural
occurrence'' in quantum theory.

The structure of the paper  follows the above discussion. In
the first section, we introduce data-driven inference as the
inference of the minimal  range consistent with the observed
distributions, and we show that the inferred range is unique
for any set  of outcome distributions only  for systems with
(hyper)-spherical state space.  We also prove that the range
of a measurement  identifies such a measurement  up to gauge
symmetries.   In  the  second   section,  we  introduce  the
property  of observational  completeness  and  show that  it
represents a strictly  stronger condition than informational
completeness.   For  systems  with  (hyper)-spherical  state
space, we show that the minimal observationally complete set
of  states happens  to  be  SIC. In  the  third section,  we
implement  a  protocol for  the  machine  learning of  qubit
measurements    based   on    data-driven   inference    and
reconstruction, and we test it  on data generated by the IBM
Q Experience quantum computer.

\section{Data-driven inference}

Let us  consider an  experimental setup involving  two boxes
equipped with $m$ buttons and $n$ light bulbs, respectively.
This situation is depicted in Fig.~\ref{fig:setup}.
\begin{figure}[tb]
  \vspace{3mm}
  \includegraphics[width=.8\columnwidth]{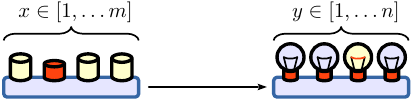}
  \caption{The   experimental   setup    consists   of   two
    uncharacterized boxes equipped with  $m$ buttons and $n$
    light bulbs,  respectively.  At each run,  Alice presses
    button  $x$  and  observes outcome  $y$,  thus  recording
    vectors  $\{ \mathbf{p}_x  \in  \R^n  \}$, whose  $y$-th
    entry is the frequency of outcome $y$ given input $x$.}
  \label{fig:setup}
\end{figure}
At each  run of the  experiment, a theoretician,  say Alice,
presses button $x$ and observes outcome $y$.  She records the
vectors $\{ \mathbf{p}_x \in \R^n \}$, whose $y$-th entry is
the frequency of outcome $y$ given input $x$.

We  address the  problem of  inferring all  the measurements
$\hat{M}$    that     are    \textit{self-consistent}    and
\textit{minimal} for  observed frequencies  $\{ \mathbf{p}_x
\}$, in  an i.i.d.  hypothesis for  $\hat{M}$.  To formalize
this  idea,  notice  that any  $n$-outcome  measurement  $M$
induces  a  linear  transformation   from  the  state  space
$\bb{S}$ (the set of all  states available to the system) to
the  space of  outcome distributions  $\R^n$.  The  range of
such  a transformation,  denoted by  $M(\bb{S})$, represents
the  distributions compatible  with measurement  $M$. Hence,
given some  prior information about a  state space $\bb{S}$,
the inferred  linear transformations $\hat{M}$  minimize the
volume\footnote{Notice  that  the  choice of  volume  as  an
  objective    function    is    independent    of    linear
  transformations, since  for any body the  volume under any
  linear  transformation  changes  by  a  factor  that  only
  depends on the linear transformation,  and not on the body
  itself.   Non-linear transformations  are not  allowed, as
  they  do   not  preserve  the  linear   structure  of  the
  underlying  space   (state/effect  space   or  probability
  space).}   of  their  ranges $\hat{M}(\bb{S})$  under  the
self-consistency  constraints $\hat{M}(\bb{S})  \supseteq \{
\mathbf{p}_x  \}$.\footnote{Notice   that  not   any  linear
  transformation  corresponds to  a legitimate  measurement.
  In the case  in which none of the inferred  $\hat{M}$ is a
  legitimate measurement, the inference  fails, in the sense
  that Alice declares that  either the data are insufficient
  or  that the  assumption of  the state  space $\bb{S}$  is
  inconsistent.}  This naturally identifies two steps in the
inference of  $\hat{M}$: i)  the inference of  the (possibly
not  unique) minimum-volume  range $\hat{\mc{R}}$  such that
$\hat{\mc{R}}  \supseteq \{  \mathbf{p}_x \}$,  and ii)  the
characterization of all the linear transformations $\hat{M}$
with given range $M (\bb{S}) = \hat{\mc{R}}$.

Let  us  start  addressing  the first  step,  that  is,  the
inference of the (possibly  not unique) minimum-volume range
$\hat{\mc{R}}$   such   that  $\hat{\mc{R}}   \supseteq   \{
\mathbf{p}_x \}$.  Since such  an inference is solely driven
by the data, we  call it \textit{data-driven inference} (see
also~\cite{PhysRevLett.118.250501,dall2018device}):

\begin{dfn}[Data-driven inference]
  \label{def:inference}
  For any set  $\{\mathbf{p}_x \}$, we denote by  $\ddi ( \{
  \mathbf{p}_x  \}  |   \bb{S}  )$  the  \textit{data-driven
    inference} map
  \begin{align}
    \label{eq:inference}
    \ddi \left(  \{ \mathbf{p}_x  \} \big| \bb{S}  \right) =
    \argmin_{\mc{R} \supseteq \left\{ \mathbf{p}_x \right\}}
    \vol \left( \mc{R} \right),
  \end{align}
  where $\vol ( \mathcal{R}  )$ denotes the Euclidean volume
  of  $\mathcal{R}$ and  the  minimization  is over  subsets
  $\mc{R}  \subseteq \mathbb{R}^n$  corresponding to  linear
  transformations of  the state  space $\bb{S}$ that  lie on
  the affine subspace generated by $\{ \mathbf{p}_x \}$.
\end{dfn}
If the  prior information  does not  specify a  single state
space $\bb{S}$,  the minimum in  Eq.~\eqref{eq:inference} is
meant to run also over all such sets.

An     intuitive      geometrical     interpretation     for
Definition~\ref{def:inference}   is   illustrated   by   the
following two examples.

First,  let $\bb{S}$  be  a  (hyper)-sphere $\Sigma$.   This
scenario  encompasses the  cases  of  classical and  quantum
bits,  where the  set $\bb{S}$  is a  one-dimensional and  a
three-dimensional sphere,  respectively.  In this  case, the
optimization  in  Eq.~\eqref{eq:inference}  is  over  affine
transformations  of  a  sphere,  that  is,  ellipsoids.   It
follows~\cite{John2014,boyd2004convex}   that    the   range
returned  by map  $\ddi$ is  unique  for any  input set  $\{
\mathbf{p}_x \}$, and one has that
\begin{align}
  \label{eq:ellipsoid}
    \ddi \left(  \left\{ \mathbf{p}_x  \right\} \;  \big| \;
    \Sigma   \right)   =   \left\{  \mvee   \left(   \left\{
    \mathbf{p}_x \right\} \right) \right\},
\end{align}
where   $\mvee$   denotes   the   minimum   volume-enclosing
ellipsoid, which  is known to  be a convex problem,  and for
which      efficient      computing      algorithms      are
available~\cite{todd2016minimum}.     This   situation    is
illustrated in Fig.~\ref{fig:inference}, left-hand side.

\begin{figure}[tb]
  \includegraphics[width=.49\columnwidth]{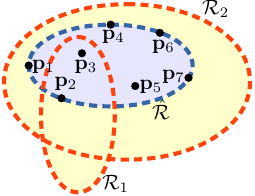}
  \includegraphics[width=.49\columnwidth]{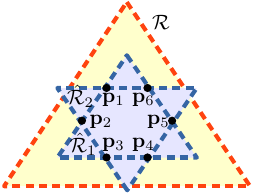}
  \caption{The figure provides a schematic representation of
    the space of distributions when $n = 3$ (three outcomes,
    and hence  the probability simplex  is two-dimensional).
    \textbf{Left}:  when  the  state  space  $\bb{S}$  is  a
    sphere,   the    data-driven   inference    defined   by
    Eq.~\eqref{eq:inference}     returns     the     minimum
    volume-enclosing ellipsoid  $\hat{\mc{R}}$, which always
    exists unique,  for the given set  $\{ \mathbf{p}_x \}$.
    Range $\mc{R}_1$ is not consistent with $\{ \mathbf{p}_x
    \}$,  while  $\mc{R}_2$,  although  consistent,  is  not
    minimum  volume.  \textbf{Right}:  when the  state space
    $\bb{S}$ is a simplex, the data-driven inference defined
    by Eq.~\eqref{eq:inference}  returns the set  of minimum
    volume-enclosing simplices  (as these are  generally not
    unique)  $\hat{\mc{R}}_1$ and  $\hat{\mc{R}}_2$ for  the
    given set  $\{ \mathbf{p}_x \}$.  Range  $\mc{R}$ is not
    minimum volume.}
  \label{fig:inference}
\end{figure}

Let us consider now, as a  second example, the case in which
the state space  $\bb{S}$ is a regular  simplex $\Delta$, as
it is  the case  for classical systems.   In this  case, the
optimization  in  Eq.~\eqref{eq:inference}  is  over  affine
transformations of  the simplex $\Delta$, which  turn out to
be  simplices  themselves.   It  is not  difficult  to  find
configurations  of the  set  $\{\mathbf{p}_x\}$ of  observed
distributions,  such  that the  range  returned  by the  map
$\ddi$  is not  unique.   This situation  is illustrated  in
Fig.~\ref{fig:inference}, right-hand side.

In light of the previous  examples, it is a natural question
to ask for which state spaces $\bb{S}$ the range returned by
the  data-driven inference  $\ddi$  is  always unique.   The
following Theorem,  proved in  the appendix, answers  such a
question.
\begin{thm}
  \label{thm:inference}
  The range returned by the data-driven inference $\ddi ( \{
  \mathbf{p}_x  \}  |  \bb{S}  )$  is  unique  for  any  $\{
  \mathbf{p}_x \}$, if and only  if the state space $\bb{S}$
  is a (hyper)-sphere.
\end{thm}

This  result can  be lifted  to  the level  of a  principle,
singling out spherical state  spaces $\bb{S}$, such as those
of the  classical and quantum  bits, as those for  which the
map $\ddi$  always returns  a unique range.   This principle
rules out theories with more exotic elementary systems, such
as PR-boxes\cite{popescu1994quantum} for  which inference is
not  always   unique.   In  this   case,  we  speak   of  an
\textit{epistemic principle},  that is, a constraint  on the
state space seen  as the hypothesis used by  the observer as
the base of their inference.

Let us now move on to  the second step mentioned above, that
is, the  characterization of all the  linear transformations
$M$ with  range $M(\bb{S})$  equal to  a given  inferred one
$\hat{\mc{R}}$.   Notice first  that any  transformation $U$
that leaves the state space $\bb{S}$ invariant, that is such
that $U  (\bb{S} )=  \bb{S}$, does not  affect the  range $M
(\bb{S})$, that is $M( U(  \bb{S})) = M (\bb{S})$.  We refer
to any  such a transformation as  a \textit{gauge symmetry}.
The  following  Theorem  shows  that accuracy  up  to  gauge
symmetries   is  indeed   the   optimal   accuracy  in   the
characterization  of  any  informationally  complete  (i.e.,
invertible) $M$, given  its range.  (For non-informationally
complete $M$, the  statement, although conceptually similar,
becomes  technically  more  involved,  so  we  postpone  the
general statement and its proof to the appendix.)

\begin{thm}
  \label{thm:range}
  For any given state  space $\bb{S}$, the range $M(\bb{S})$
  of any  informationally complete $M$ identifies  $M$ up to
  gauge symmetries.
\end{thm}

Although  Theorem~\ref{thm:range}  is  valid for  any  state
space $\bb{S}$, for the sake  of concreteness let us revisit
our  running  example  where $\bb{S}$  is  a  (hyper)-sphere
$\Sigma$, in particular a qubit.

In  the  qubit  case,  the gauge  symmetries  correspond  to
unitary  and  anti-unitary  transformations in  the  Hilbert
space,  hence,  the   range  $M(\bb{S})$  identifies  linear
transformation   $M$   up   to  unitary   and   anti-unitary
transformations.    However,  Theorem~\ref{thm:range}   only
guarantees the existence of  such an identification, without
providing         an          explicit         construction.
Reference~\cite{PhysRevLett.118.250501}  fills  this gap  by
explicitly  deriving  all  the linear  transformations  that
correspond to  any given (hyper)-ellipsoidal range.   Such a
result is provided a new simpler proof in appendix.

An  algorithmic representation  of data-driven  inference is
given in Fig.~\ref{fig:flux}.

\begin{figure}
  \includegraphics{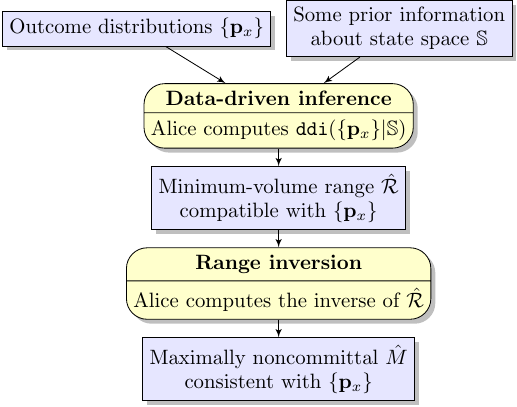}

  \caption{Algorithmic    representation   of    data-driven
    inference.}
  \label{fig:flux}
\end{figure}

\section{Data-driven  reconstruction}

In  the  previous  section  we  introduced  a  principle  of
self-consistent  minimality  to  guide the  inference  of  a
measurement    given    a    set   of    observed    outcome
distributions.  In this  section,  we consider  the case  in
which    the   boxes    with   buttons    and   lights    in
Fig.~\ref{fig:setup}  describe  a physical  state-preparator
$\mc{S}$ and a physical measurement $M$, respectively.  This
situation is illustrated in Fig.~\ref{fig:setup2}.
\begin{figure}[tb]
  \vspace{3mm}
  \includegraphics{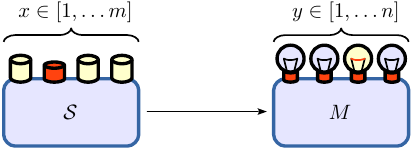}
  \caption{The    same    experimental     setup    as    in
    Fig.~\ref{fig:setup},  although  this   time  the  state
    preparator $\mc{S}$  is built  by Bob,  with the  aim of
    enabling  Alice  to  correctly infer  measurement  range
    $M(\bb{S})$,   that   is,   for   the   inferred   range
    $\hat{\mc{R}}$, one has $\hat{\mc{R}}=M(\bb{S})$.}
  \label{fig:setup2}
\end{figure}

The    state-preparator   $\mc{S}$    is    built   by    an
experimentalist, say Bob, with the  aim of enabling Alice to
correctly  infer the  measurement range  $M(\bb{S})$ through
DDI, that  is, the inferred range  $\hat{\mc{R}}$ satisfying
$\hat{\mc{R}}=M(\bb{S})$,  in  the   limit  in  which  Alice
presses each button infinitely many times.  This task shares
similarities with conventional  measurement tomography, with
the major difference that, in  the latter, full knowledge of
the state-preparator  $\mc{S}$ is  pivotal for  Alice's data
analysis, whereas DDI solely depends on the observed outcome
distributions and the knowledge of the state space $\bb{S}$.

Since it is sufficient to show that Bob is able to construct
one such a state-preparator, we  can assume, without loss of
generality, that  each button of the  state preparator emits
always  the same  state at  each press,  and that  different
buttons are  associated with  different states.   Hence, the
state-preparator   $\mathcal{S}$   can   be   mathematically
described as a set of states.

The probabilities  $\{ \mathbf{p}_x  \}$ Alice  observes are
the  image  of   the  states  in  $\mc{S}$,   that  is,  $\{
\mathbf{p}_x\} = M(\mc{S})$.  Correct inference imposes then
that  $M(\mc{S})$ contains  all the  statistical information
that is available in the measurement range $M (\bb{S})$, and
that  such  information  can  be  extracted  by  data-driven
inference.  We  call \textit{observationally  complete} (OC)
any state  preparator that allows for  the correct inference
of measurement $M$.

\begin{dfn}[Observational completeness]
  \label{def:completeness}
  A    set   of    states   $\mc{S}$    is   said    to   be
  \textit{observationally  complete}   for  measurement  $M$
  whenever
  \begin{align}
    \label{eq:completeness}
    \ddi  \left(  M(\mc{S})\big|   \bb{S}
    \right)= \left\{ M (\bb{S}) \right\}\;.
  \end{align}
\end{dfn}

If the set  is not OC, the set of  inferred measurements may
not contain  the correct measurement, but  measurements with
ranges  smaller  than  the   correct  measurement,  or  even
measurements inequivalent (modulo gauge symmetries) to it.

Notice  that observational  completeness (OC-ness)  plays an
analogous   role    for   data-driven    reconstruction   as
informational completeness (IC-ness)  plays for conventional
measurement tomography.   However, IC  sets of  states allow
for the  correct tomographic reconstruction  of \textit{any}
measurement, while  OC sets  of states apparently  depend on
the measurement to be inferred.

Is this really  the case? It turns out that,  as long as the
measurement is  IC, by  bypassing the  linear transformation
$M$ in  Eq.~\eqref{eq:completeness} one obtains  a condition
equivalent to Eq.~\eqref{eq:completeness},  as stated in the
following theorem:
\begin{thm}
  \label{thm:completeness}
  A  set  $\mc{S}$ of  states  is  OC  complete for  any  IC
  measurement, if and only if
  \begin{align}
    \label{eq:completeness2}
    \ddi \left( \mc{S} \big|  \bb{S}\right) = \left\{ \bb{S}
    \right\}.
  \end{align}
\end{thm}
Notice that  in Eq.~\eqref{eq:completeness2} the  map $\ddi$
is  applied to  a set  of  states, as  opposed to  a set  of
probability distributions as it was the case so far.

Hence, any  set $\mc{S}$ of  states that  is OC for  some IC
measurement,   is  also   OC  for   \textit{any  other}   IC
measurement.      Moreover,     Eq.~\eqref{eq:completeness2}
provides a  characterization of any  such a set  $\mc{S}$ in
closed-form,  namely,  in  a  form  which  only  depends  on
$\mc{S}$        alone,        in        contrast        with
Definition~\ref{def:completeness} that also  depends on $M$.
Such a characterization can be  used to readily check if any
given set $\mc{S}$ is OC for any IC measurement.

More generally,  even if $M$  is not an IC  measurement, one
can      write      a      condition      equivalent      to
Eq.~\eqref{eq:completeness}  (and conceptually  analogous to
Eq.~\eqref{eq:completeness2},    just    technically    more
involved)  that depends  on  $M$ only  through its  support.
Hence, any set $\mc{S}$ of states that is OC for $M$ is also
OC for any other measurement with the same support. In other
words, it is  \textit{universally} OC on such  a support.  A
fully general  version of  Theorem~\ref{thm:completeness} is
provided in the appendix.

Notice that, while OC sets  of states are universal within a
given  subspace, IC  sets of  states are  universal for  any
subspace and \textit{all}  of its subspaces. In  fact, it is
easy  to see  that the  only set  of states  that is  OC for
\textit{any} measurement  $M$ is  trivially the  state space
$\bb{S}$.

\begin{figure}[htb]
  \includegraphics{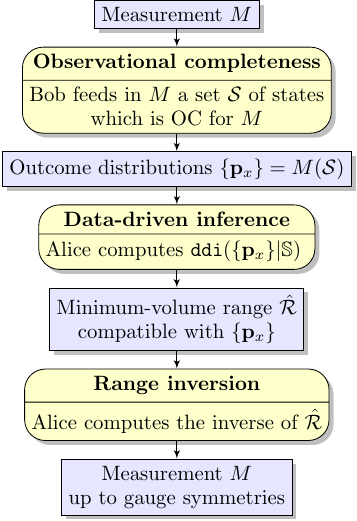}
  \caption{Algorithmic    representation   of    data-driven
    reconstruction.}
  \label{fig:flux2}
\end{figure}

Although  Theorem~\ref{thm:completeness}  is valid  for  any
state space  $\bb{S}$, for the  sake of concreteness  let us
revisit   our   running   example  where   $\bb{S}$   is   a
(hyper)-sphere.

In         light        of         Eq.~\eqref{eq:ellipsoid},
Theorem~\ref{thm:completeness} states that a set $\mc{S}$ of
states is OC for IC measurements  if and only if its minimum
volume-enclosing ellipsoid coincides with the (hyper)-sphere
$\Sigma$, that is
\begin{align}
  \mvee(\mc{S}) = \Sigma.
\end{align}
In     the     appendix,     as     a     consequence     of
Ref.~\cite{vandev1992minimal}, we show that such a condition
is  satisfied  when $\mc{S}$  is  a  regular simplex.   This
situation  is illustrated  in Figure~\ref{fig:completeness},
left-hand side.  Moreover, we show that, as a consequence of
Ref.~\cite{Vince2008},  also the  converse is  true, namely,
that such  a condition is  violated whenever $\mc{S}$  is an
irregular   simplex    (see   Figure~\ref{fig:completeness},
right-hand   side).   Hence,   for  simplices,   OC-ness  is
equivalent to symmetric IC-ness  and, therefore, the SIC set
of states is  the minimal (in terms of  cardinality) OC set.
This  provides an  operational  interpretation of  symmetric
informational completeness in terms of data-driven inference
and reconstruction.  It is  tempting to conjecture that this
equivalence  holds  for any  quantum  system,  not just  the
qubit.

\begin{figure}[b]
  \includegraphics[width=.33\columnwidth]{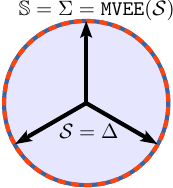}
  \hspace{.1\columnwidth}
  \includegraphics[width=.33\columnwidth]{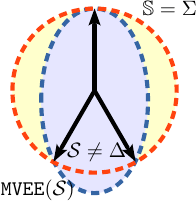}
  \caption{The    figure    provides    a    two-dimensional
    representation  of the  spherical  state space  $\bb{S}$
    (orange  line)  and   of  the  minimum  volume-enclosing
    ellipsoid (blue line) for the simplicial set $\mc{S}$ of
    states  (black vectors).   \textbf{Left}: when  $\mc{S}$
    coincides  with   the  regular  simplex   $\Delta$,  its
    minimum-volume  enclosing ellipsoid  coincides with  the
    sphere        $\Sigma$,       hence,        due       to
    Theorem~\ref{thm:completeness},        $\mc{S}$       is
    observationally complete.  \textbf{Right}: when $\mc{S}$
    does not coincide with the regular simplex $\Delta$, its
    minimum  volume-enclosing  ellipsoid  does  \textit{not}
    coincide  with  the  sphere   $\Sigma$,  hence,  due  to
    Theorem~\ref{thm:completeness},    $\mc{S}$    is    not
    observationally complete.}
  \label{fig:completeness}
\end{figure}

Although it is  easy to see that any set  of states which is
OC on a subspace, is also  IC on that subspace, the previous
example shows that the vice-versa  is not true.  Indeed, any
simplex,  whether regular  or not,  is trivially  IC on  its
support.  Hence,  OC-ness is  a strictly  stronger condition
than IC-ness.  In this sense, OC-ness defines a new ``Bureau
of  Standards''  in  terms   of  data-driven  inference  and
reconstruction.

\section{Machine larning of qubit measurements}

In  this Section,  we  first provide  an  algorithm for  the
data-driven  inference of  measurements of  any system  with
(hyper)-spherical state space,  and then we apply  it to the
data-driven   reconstruction   of    a   qubit   measurement
implemented with the IBM quantum computer.

Let us begin by discussing the algorithm for the data-driven
inference.   As   shown  by   Eq.~\eqref{eq:ellipsoid},  for
(hyper)-spherical   state  space   $\bb{S}  =   \Sigma$  the
data-driven  inference algorithm  $\ddi$ can  be written  in
terms  of a  minimum volume  enclosing ellipsoind  algorithm
$\mvee$. Efficient algorithms for the computation of $\mvee$
(quadratic  in the  dimension and  linear in  the number  of
points) can  readily be  found~\cite{todd2016minimum, TY07},
but in general assume that the affine space generated by the
input  set coincides  with  the full  linear space.   Hence,
given  a  set  $\{  \mathbf{p}_x \}$,  to  compute  $\ddi(\{
\mathbf{p}_x \} | \Sigma)$ one proceeds as follows:
\begin{enumerate}
\item Find an isometry $V$ such  that $V^T V = \openone$ and
  $V  V^T$   is  the   projector  on  the   linear  subspace
  homogeneous  to  the  affine  subspace  generated  by  $\{
  \mathbf{p}_x \}$.
  \item Compute  $\mvee( V^T \{ \mathbf{p}_x  - \mathbf{p}_0
    \})$,   where  by   construction  the   affine  subspace
    generated by $V^T \{  \mathbf{p}_x - \mathbf{p}_0 \}$ is
    now full dimensional, thus obtaining an ellipsoid in the
    form   $(\tilde{\mathbf{p}}    -   \tilde{\mathbf{t}})^T
    \tilde{Q}^{-1} (\tilde{\mathbf{p}} - \tilde{\mathbf{t}})
    \le 1$.
  \item The output of $\ddi(\{ \mathbf{p}_x \} | \Sigma)$ is
    the singleton given by the ellipsoid
    \begin{align*}
    \begin{cases}
      \left( \openone  - Q  Q^+ \right) \left(  \mathbf{p} -
      \mathbf{t}   \right)  =   0,\\  \left(   \mathbf{p}  -
      \mathbf{t}   \right)^T   Q^+   \left(   \mathbf{p}   -
      \mathbf{t} \right) \le 1,
    \end{cases}
    \end{align*}
    where  $Q  := V  \tilde{Q}  V^T$  and $\mathbf{t}  :=  V
    \tilde{\mathbf{t}} + \mathbf{p}_0$.
\end{enumerate}

We  tested  our  algorithm  on  the  set  $\{  \mathbf{p}_0,
\mathbf{p}_1, \mathbf{p}_2 \}$ as  follows (later on we show
how this  data was  obtained):
\begin{align}
  \label{eq:data}
  \mathbf{p}_0 \simeq
  \begin{bmatrix}
    0.48 \\ 0.05 \\ 0.24 \\ 0.23
  \end{bmatrix}, \quad
  \mathbf{p}_1 \simeq
  \begin{bmatrix}
    0.21 \\ 0.32 \\ 0.42 \\ 0.05
  \end{bmatrix}, \quad
  \mathbf{p}_1 \simeq
  \begin{bmatrix}
    0.17 \\ 0.33 \\ 0.11 \\ 0.39
  \end{bmatrix}
\end{align}
The minimum-volume enclosing ellipsoid for the set $\{ V^T (
\mathbf{p}_x -  \mathbf{p}_0 ) \}$,  as given in  the second
step of our algorithm, is depicted in Fig.~\ref{fig:mvee}.
\begin{figure}[h!]
  \includegraphics[width=\columnwidth]{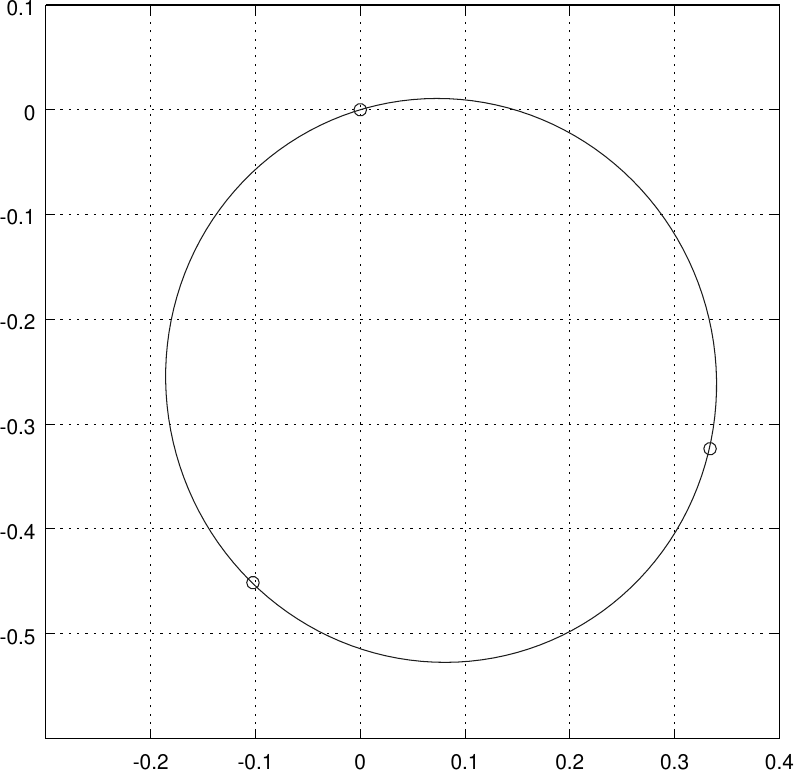}
  \caption{The  figure  depicts  the  three  points  $V^T  (
    \mathbf{p}_x - \mathbf{p}_0 )$ (circles)  for $x = 0, 1,
    2$, where $\bf{p}_x$'s are given by Eq.~\eqref{eq:data},
    and their minimum-volume  enclosing ellipsoid, according
    to the second step of  our algorithm for the data-driven
    inference of qubit measurements.}
  \label{fig:mvee}
\end{figure}
The  correlation   matrix  $Q$   and  the  vector   $t$,  as
reconstructed by the third step  of our algorithm, are given
by
\begin{align}
  \label{eq:corrmatrix}
  Q \simeq 10^{-2} \begin{bmatrix}
    3.8 & -3.6 & 0 & -0.1\\
    -3.6 & 3.4 & 0.2 & 0\\
    0 & 0.2 & 3.3 & -3.4\\
    -0.1 & 0 & -3.4 & 3.6
  \end{bmatrix},
  \quad       \mathbf{t}        \simeq       \begin{bmatrix}
    0.29\\ 0.23\\ 0.26\\ 0.22
  \end{bmatrix}.
\end{align}
As a  consequence of Theorem~\ref{thm:range},  through range
inversion this  allows for the data-driven  inference of the
POVM up to gauge symmetries,  that is, an overall unitary or
anti-unitary  transformation.  In  particular,  as shown  in
Ref.~\cite{PhysRevLett.118.250501}  (an   alternative,  more
compact proof is given in the appendix), any measurement $\{
\pi_y \}$ compatible with  correlation matrix $Q$ and vector
$t$ can be obtained by inverting the following system:
\begin{align*}
  \begin{cases}
    Q_{x,y}  =   \frac12  \Tr[  \pi_x  \pi_y   ]  -  \frac14
    \Tr[\pi_x] \Tr[\pi_y],\\ t_x = \frac12 \Tr[\pi_x].
  \end{cases}
\end{align*}

Let us now discuss  the data-driven reconstruction, that is,
the  way  in  which  the  data  in  Eq.~\eqref{eq:data}  was
generated. Such data  was generated by the  IBM Q Experience
quantum  computer.   The  ideal   circuit  consists  of  the
preparation of a trine set  $\{ \phi_x \}$ of states, which,
as     observed     in    Fig.~\ref{fig:completeness}     is
observationally complete  for any real measurement,  and the
measurement of  two mutually unbiased basis,  as depicted in
Fig.~\ref{fig:circuit}.
\begin{figure}[h!]
  \includegraphics[width=0.9\columnwidth]{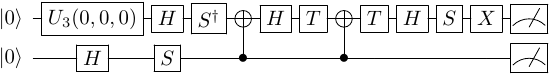}
  \\\phantom{X}\\
  \includegraphics[width=0.9\columnwidth]{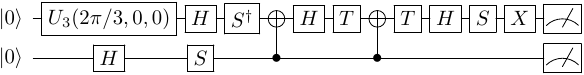}
  \\\phantom{X}\\
  \includegraphics[width=0.9\columnwidth]{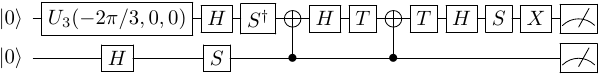}
  \caption{The  orange   gate  (top  left)   represents  the
    preparation of  the trine ensemble (from  top to bottom,
    states  $\phi_0$,  $\phi_1$  and $\phi_2$).   All  other
    gates  represent a  Naimark extension  of a  measurement
    consisting of two mutually unbiased basis.}
  \label{fig:circuit}
\end{figure}
This allows  one to compare  the correlation matrix  $Q$ and
the vector  $\mathbf{t}$ given  in Eq.~\eqref{eq:corrmatrix}
(from which one  can recover the inferred  measurement up to
gauge  symmetries through  range inversion)  with the  ideal
ones, given by:
\begin{align*}
  Q = \frac1{16} \begin{pmatrix} 1 & -1 & 0 & 0\\ -1 & 1 & 0
    & 0\\ 0 & 0 & 1 & -1\\ 0 & 0 & -1 & 1
  \end{pmatrix}, \quad \mathbf{t} = \frac14 \begin{pmatrix} 1\\1\\1\\1 \end{pmatrix}.
\end{align*}
An  analysis  of the  noise  affecting  the IBM  backend  is
however out of the scope of this manuscript.

\section{Conclusion}

In this work  we introduced data-driven inference  as a rule
to output the  maximally noncommittal measurement consistent
with a  set of observed  distributions.  We showed  that the
inference is  possible in principle up  to gauge symmetries,
that  is, symmetries  of the  state space  of the  system at
hand,  and   that  this  accuracy  limit   is  achieved  for
(hyper)-spherical state spaces. Then, we considered the task
of reconstructing  an unknown measurement via  DDI.  To this
aim, we introduced observationally  complete sets of states,
as those enabling  a correct inference \textit{universally},
that  is, for  \textit{any} unknown  measurement on  a given
support.    Deriving  a   closed-form  characterization   of
observational completeness  allowed us  to show  that, while
observational completeness is  a strictly stronger condition
than   informational   completeness,    in   the   case   of
(hyper)-spherical state space OC-ness with minimum number of
states is equivalent to SIC-ness.  In this way, the protocol
of  data-driven inference  provides SIC  sets with  a novel,
entirely operational interpretation.

\section*{Acknowledgements}

M.D.  was  supported by  MEXT Quantum Leap  Flagship Program
(MEXT Q-LEAP) Grant No.  JPMXS0118067285, JSPS KAKENHI Grant
Number JP20K03774,  and the  International Research  Unit of
Quantum Information,  Kyoto University.   F.B.  acknowledges
support  from  MEXT  Quantum  Leap  Flagship  Program  (MEXT
Q-LEAP), Grant  Number JPMXS0120319794,  and from  the Japan
Society for the Promotion  of Science (JSPS) KAKENHI, Grants
Number 19H04066  and Number 20K03746.  This  work was partly
supported by the program  for FRIAS-Nagoya IAR Joint Project
Group.   A.~B.  and  A.~T.  acknowledge  the John  Templeton
Foundation, Project ID 60609 Quantum Causal Structures.

\setcounter{dfn}{0}
\setcounter{thm}{0}
\setcounter{cor}{0}

\appendix

\section*{Appendix}

While, for the sake of  clarity, in the previous sections we
focused  on the  data-driven inference  of measurements,  in
these appendices  we extend  the formalism to  encompass the
case of  data-driven inference  of families of  states, thus
justifying  the word  ``devices''  in  our title.   Finally,
rather  than restricting  the  presentation  to the  quantum
case, here we consider general probabilistic theories.

\section{General framework}

A physical system  can be defined by giving a  set of states
and  a   set  of  effects,  representing   respectively  the
preparations and the observations  of the system.  An effect
$a \in \mathbb{E}$ is a linear map that takes a state $\rho$
as  an input  and outcomes  a probability  $p(a |  \rho )  :=
a(\rho)$.   Since  randomization of  different  experimental
setups is in itself another  valid experiment, it is natural
to endow  the set of  states and the  set of effects  with a
linear  structure and  allow for  any convex  combination of
states and effects. By linear  extension, it is also natural
to introduce  the real vector  spaces generated by  any real
linear combinations  of states and effects.   Restricting to
the finite dimensional case, the  linear space of states and
the linear space of effects are  dual to each other and both
isomorphic  to $\mathbb{R}^{\ell}$  for some  natural number
$\ell$ which  is called  the \textit{linear dimension}  of the
system.

We  assume that  the  physical theory  is \textit{causal}.   A
probabilistic  theory is  causal  if there  exists a  unique
\textit{deterministic   effect}   $e  \in   \mathbb{E}$,   and
\textit{deterministic  states}  are  those  states  such  that
$e(\rho)=1$.  Therefore, states can  always be normalized as
$\overline{\rho}  :=   \rho/e(\rho)$  and  every   state  is
proportional to  a deterministic one.  For  this reason, the
full  set  of states  of  any  causal theory  is  completely
specified by the set of  deterministic states.  We denote by
$\mathbb{S}$ the set of normalized (or deterministic) states
of the theory and by $\mathbb{E}$  the set of effects of the
theory.

By  choosing an  arbitrary basis,  we can  give a  geometric
representation of $\mathbb{S}$ and $\mathbb{E}$ as subset of
$\mathbb{R}^{\ell}$:  $\mathbb{S}$  will  be  a  convex  set
contained  in   a  strictly   affine  $(\ell-1)$-dimensional
subspace, while $\mathbb{E}$ will  be a ``bicon-ish'' shaped
solid (see Fig.~\ref{fig:stateseffects}).
\begin{figure}[h!]
  \includegraphics[width=\columnwidth]{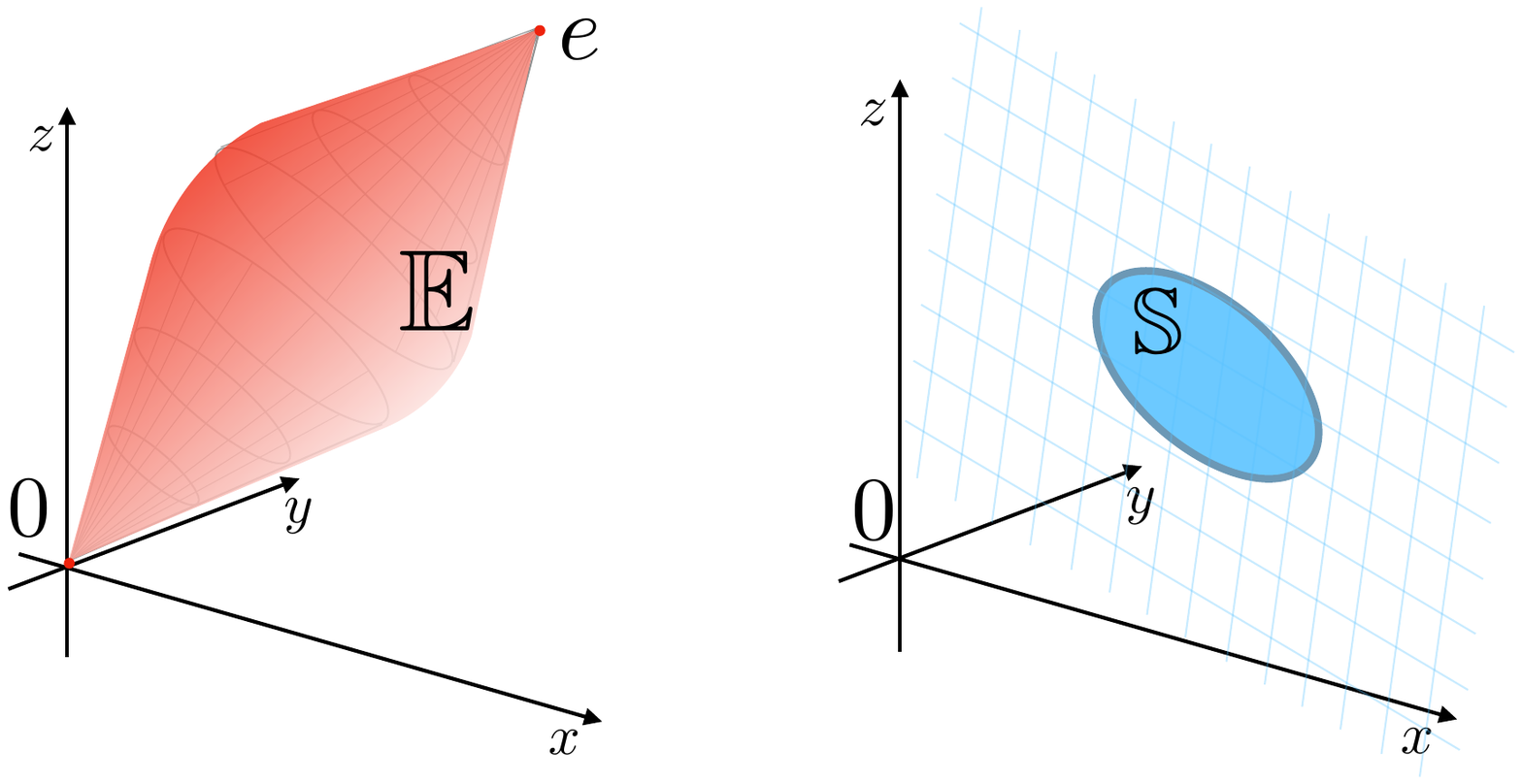}
  \caption{Vectorial   representation   of  the   a   system
    $(\bb{S},\bb{E})$  of linear  dimension $\ell=3$.   Both
    the set  of normalized  states $\bb{S}$  and the  set of
    effects $\bb{E}$ can by represented as convex subsets of
    the vector space $\R^3$.  { \bf Left:} The convex set of
    effects $\bb{E}$  corresponds to a truncated  cone.  The
    truncated  cone contains  the two  points $(0,0,0)$  and
    $(1,1,1)$, which  correspond to  the null effect  and to
    the  deterministic  effect   $e$  respectively.   {  \bf
      Right:} The convex set  of normalized states $\bb{S}$.
    This   set  contains   all  states   $\rho$  such   that
    $e(\rho)=1$, and is then merged  in an affine space with
    dimension  $\ell-1$  (in  this  case  a  two-dimensional
    space).}
  \label{fig:stateseffects}
\end{figure}

\textit{Measurements}    are     a    family     of    effects
$\{a_y\}_{y=1}^{n}$  such  that   $\sum_{y=1}^{n}a_y  =  e$,
$\forall \rho\in\mathbb{S}$. Since any state can be regarded
as a vector  in $\R^\ell$, any measurement  induces a linear
map $M\in\M{n}{\ell}$ defined as follows
\begin{align*}
  M:&\,\R^l\rightarrow \R^n\\
      & \rho _x \mapsto \mathbf{p}_x, \qquad  
  p_x^y := a_y(\rho_x).
\end{align*}
(each row  of $M\in\M{n}{\ell}$  corresponds to  an effect).
Any state  is mapped  into a  point in  $\R^n$ (see  the Top
Fig.~\ref{fig:linearmap}).   Analogously,  since any  effect
$a_x$  corresponds to  a  vector  in $\mathbb{R}^\ell$,  any
family of  states $\{ \rho_y  \}_{y =1}^n$ induces  a linear
map defined as follows
\begin{align*}
  R:\,&\mathbb{R}^{l} \to \mathbb{R}^{n} \\
    &a_x \mapsto \mathbf{p}_x,  \qquad
      {p}^y_x := \rho_y(a_x).
\end{align*}
(each row of $R\in\M{n}{\ell}$ corresponds to a state).  Any
effect  is mapped  in  a  point in  $\R^n$  (see the  Bottom
Fig.~\ref{fig:linearmap}).

\begin{figure}[t!]
  \includegraphics[width=\columnwidth]{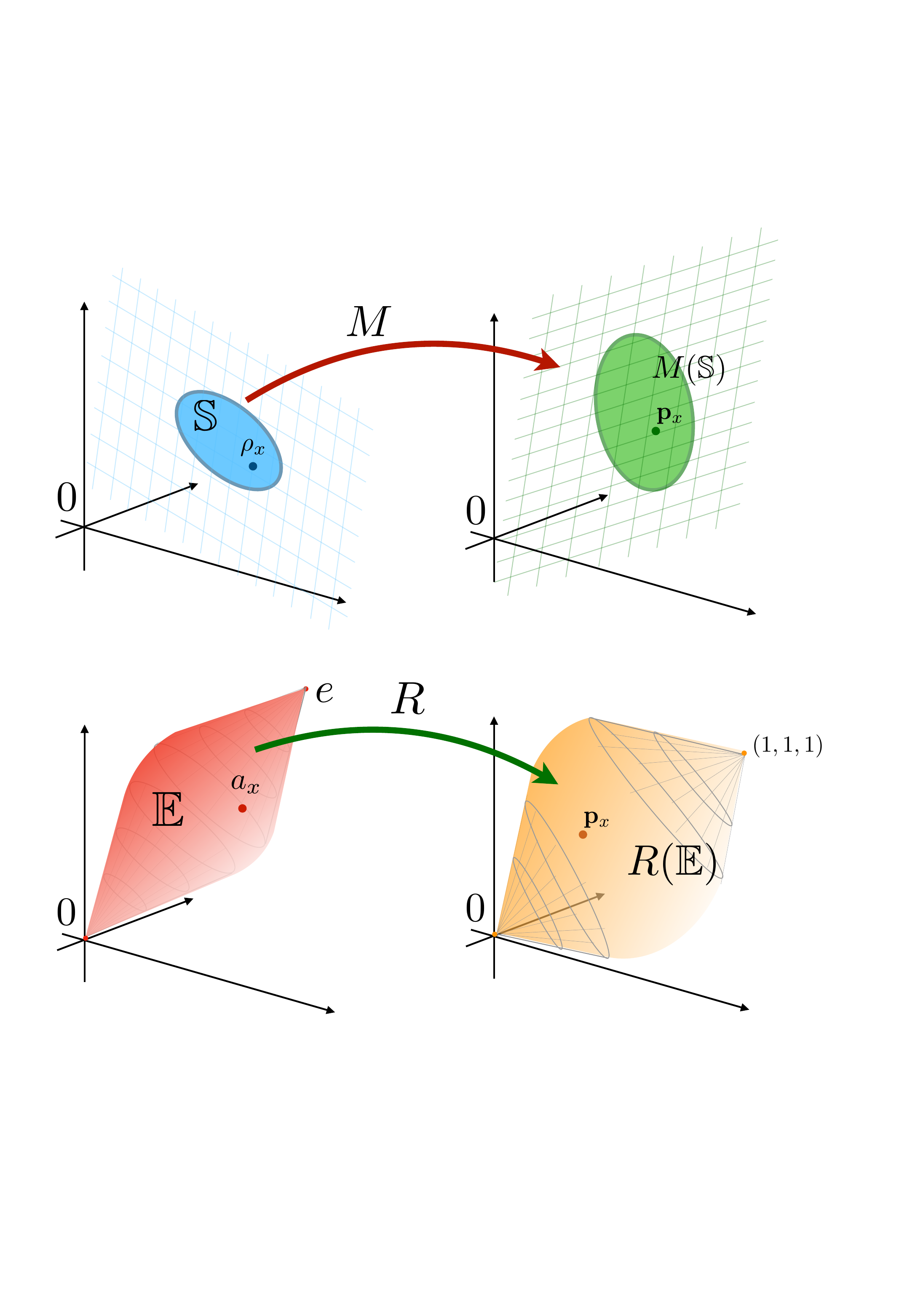}\\
  \caption{{ \bf Top}: The linear map $M$ corresponding to a
    measurement.   {\bf   Bottom}:   The  linear   map   $R$
    corresponding to a family of states.}
\label{fig:linearmap}
\end{figure}

For  example, in  quantum theory,  any system  is associated
with a  $d$-dimensional Hilbert space  $\mathcal{H}$, states
and  effects  are   represented  by  positive  semi-definite
operators  on $\mathcal{H}$,  and conditional  probabilities
are given  by the Born  rule: $a_y(\rho_x):=\Tr[a_y\rho_x]$.
States and  effects are represented  by vectors in  the real
space $\R^{\ell}$ with $\ell=d^2$.  As an example, any qubit
($d=2$)        normalized        state        $\rho        =
\tfrac{1}{2}\sum_{i=0}^{3}r^i\sigma_i$     is     univocally
associated to a  vector $(r^0,r^1,r^2,r^3)^T\in\R^4$.  Here,
the condition  $r^0 =  1$ guarantees  that $\Tr[\rho]  = 1$,
while  the   condition  $|r|   \leq  1$  for   Bloch  vector
$\mathbf{r}        :=         (r^1,r^2,r^3)^T$,        where
$\{\sigma_i\}_{i=1}^{3}$ are the  Pauli matrices, guarantees
that $\rho \ge 0$. The  set of normalized states, identified
by the constraint $r^0=1$, is geometrically represented by a
sphere  (the Bloch  sphere) contained  in a  $3$-dimensional
strictly affine subspace of $\mathbb{R}^4$.

As  a  second  example,  in  classical  theory  the  set  of
normalized states of any system with linear dimension $\ell$
is  an $\ell$-simplex,  e.g.  a  segment for  the \textit{bit}
system $\ell=1$  and a  triangle for the  \textit{trit} system
$\ell=2$.

In  the  literature, toy  models  have  been proposed  whose
convex set  of states  (and effects)  differs from  both the
quantum and  the classical  ones.  The most  notable example
are  the PR-boxes,  whose  convex set  of  states of  linear
dimension $\ell=3$ is a square.

\section{Data-driven inference}

The   \textit{data-driven   inference}    (DDI)   of   quantum
measurements, presented in the Main Text, is a protocol that
allows  to  infer  a  preferred (according  to  a  maximally
noncommittal  criterion)  measurement  from a  set  of  data
interpreted as  the outcome  distributions of  an experiment.
The maximally noncommittal criterion  is very natural: among
the   set  of   measurements   whose   range  includes   the
experimental  points, we  choose  those with  minimum-volume
range  in  the space  of  outcome  distributions. Hence,  the
algorithmic idea  at the basis of  the data-driven inference
of quantum measurements is very simple: i) the first step is
the  search of  the  minimum-volume enclosing  ranges for  a
given set  of points, ii) the  second step is the  search of
the measurements that  are able to reproduce  such ranges as
the outcome distributions of an experiment.

It  is intuitive  that in  an analogous  way one  can define
data-driven inference  of quantum states.  Moreover,  due to
its  purely  geometrical  nature, the  idea  of  data-driven
inference of physical devices is not anchored to the quantum
formalism,  but  can be  defined  in  the  same way  in  any
possible  probabilistic  theory.  Within  this  perspective,
here we define the data-driven inference of measurements and
of states in the framework of general probabilistic theories
\cite{PhysRevA.75.032110,                popescu1994quantum,
  PhysRevA.75.032304,PhysRevA.71.022101,   hardy2001quantum,
  Popescu:2014aa,                              d2010testing,
  dariano_chiribella_perinotti_2017}.

\subsection{Inference of measurements}

A setup  comprising two boxes,  one equipped with  $n$ bulbs
and the other equipped with  $m$ buttons, is given.  At each
run  of the  experiment a  theoretician, say  Alice, presses
button  $x$  and records  which  bulb  $y$ lights  up.   She
iterates   this   procedure   many  times,   recording   the
frequencies  $\{ \mathbf{p}_x  \}$ whose  $y$-th element  is
the probability of  outcome $y$ given input  $x$. This situation
is     illustrated     in      the     upper     part     of
Fig.~\ref{fig:experimental-setups}.

\begin{figure}[t!]
  \includegraphics[width=.8\columnwidth]{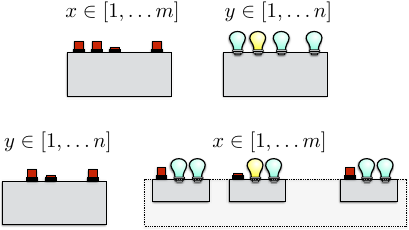}
  \caption{{\bf   Top}:   Setup   for   the   inference   of
    measurements.  {\bf Bottom}: Setup  for the inference of
    families of states.  }
  \label{fig:experimental-setups}
\end{figure}

The aim  of data-driven  inference is  to infer  the maximal
noncommittal   linear  maps   $M  :   \R^l\rightarrow  \R^n$
consistent with  $\{ \mathbf{p}_x  \}$, that is,  the linear
map $M$  that minimize the  volume of $M(\bb{S})$  such that
the  range   $M(\bb{S})$  contains  the   distributions  $\{
\mathbf{p}_x\}\subseteq\R^n$. We recall that $\bb{S}$ is the
set of all states of  the system of linear dimension $\ell$.
As  already  noticed  in  the  Main  Text,  not  any  linear
transformation corresponds  to a legitimate  measurement: in
case of  a non physical  inference $M$, failure  is declared
and a  larger set  $\{ \mathbf{p}_x  \}$ is  required.  This
definition of the problem naturally identifies two steps: i)
inferring  the (possibly  non  unique) minimum-volume  range
$\hat{\mc{R}}$ consistent with $\{ \mathbf{p}_x \}$, and ii)
finding   the   measurements   $\hat{M}$  whose   range   is
$\hat{\mc{R}}$.

In the following definition we  formalize the first of these
two steps, that  is, the inference process  that consists of
finding  the  minimum-volume  range  $M(\bb{S})$  such  that
$M(\bb{S})  \supseteq \{  \mathbf{p}_x \}$.

\begin{dfn}[Data-driven inference of measurements]
  \label{def:ddi-measurements}
  For  any set  $\{\mathbf{p}_x\}\subseteq\R^n $,  we denote
  with $\ddi  ( \{  \mathbf{p}_x \}  \; |  \; \bb{S}  )$ the
  \textit{data-driven inference} map
  \begin{align}
    \label{eq:ddi-measurements}
    \ddi \left(  \left\{ \mathbf{p}_x  \right\} \;  \big| \;
    \bb{S}  \right)  =   \argmin_{\substack{M(\bb{S})\\\{\mathbf{p}_x\}  \subseteq
    M(\bb{S}) \subseteq  \aff \{\mathbf{p}_x\}}}
    \vol\left(M(\bb{S})\right),
  \end{align}
  where $\vol ( \mathcal{R}  )$ denotes the Euclidean volume
  of  $\mathcal{R}$ and  the  minimization  is over  subsets
  $M(\bb{S})$ corresponding to linear transformations of the
  set of  states $\bb{S}$  that lie  on the  affine subspace
  generated by $\{\mathbf{p}_x\}$.
\end{dfn}

In general the outcome of the  DDI map is not unique, and the
map  returns  a  set  of  ranges. We  notice  that,  if  the
available prior  information does not identify  a unique set
$\bb{S}$  of   states,  the  set  $\bb{S}$   itself  can  be
considered as part of the  optimization problem in the above
definition,      by      taking     the      minimum      of
Eq.~\eqref{eq:ddi-measurements} over any possible $\bb{S}$.

We  can now  derive the  main result  of this  section, that
establishes the  special role played  by (hyper)-ellipsoidal
sets  $\bb{S}$  of states  in  the  context  of the  DDI  of
measurements.

Before  stating  the  main  theorem we  need  the  following
definition.

\begin{dfn}[$\mc{U}$-symmetric set]
  Given a set of transformations $\mc{U}$, and a set $\mc{X}
  \subseteq    \R^k$,    we    say    that    $\mc{X}$    is
  $\mc{U}$-symmetric  if  $U(\mc{X})  :=  \{  U\mathbf{x}  |
  \mathbf{x} \in  \mc{X} \} = \mc{X}$  for any $U\in\mc{U}$,
  where we  have chosen a $k$-dimensional  representation of
  the set $\mc{U}$.
\end{dfn}

\begin{thm}
  \label{thm:inference-sm}
  Given a  set of  states $\bb{S}$ the  following conditions
  are equivalent:
  \begin{enumerate}
  \item\label{item:singleton-sm} $\ddi(\mc{X}|\bb{S})$  is a
    singleton for any $\mc{X} \subseteq \R^n$ and any $n \in
    \N$.
  \item\label{item:symmetry-sm}   $\ddi(\mc{X}|\bb{S})$   is
    $\mc{U}$-symmetric for any $\mc{U}$--symmetric $\mc{X}$,
    where $\mc{U}$ is a set of orthogonal matrices.
  \item\label{item:ellipsoid-sm}      $\bb{S}$     is      a
    (hyper)--ellipsoid.
  \end{enumerate}
\end{thm}
\begin{proof}
  Let us prove each implication separately:
  
  \noindent 
  $\ref{item:singleton-sm} \Rightarrow \ref{item:symmetry-sm}$

  \noindent
  Consider  a $\mc{U}$--symmetric  $\mc{X}$, and  suppose by
  absurd that $\ddi(\mc{X}|\bb{S})$ is a singleton $\mc{Y}$,
  but  $\mc{Y}$  is  not  $\mc{U}$--symmetric.   Then  there
  exists  $U  \in  \mc{U}$  such   that  $\mc{Y}'  :=  \{  U
  \mathbf{y} \; | \; \mathbf{y}  \in \mc{Y} \} \neq \mc{Y}$.
  Since  $\mc{Y}' \supseteq  \mc{X}$ and  $\vol (\mc{Y}')  =
  \vol   (\mc{Y})$,  one   has  the   absurd  $\mc{Y}'   \in
  \ddi(\mc{X}|\bb{S})$.

  \noindent
  $\ref{item:symmetry-sm} \Rightarrow \ref{item:ellipsoid-sm}$
  
  \noindent      The      implication      follows      from
  Lemma~\ref{lmm:affine}, by taking $\mc{X}$  to be a sphere
  and   observing  that   $\{M(\bb{S})\,|\,M\in\M{n}{\ell})$
  contains  all  spheres  if  and  only  if  $\bb{S}$  is  a
  (hyper)--ellipsoid.
  
  \noindent
  $\ref{item:ellipsoid-sm} \Rightarrow \ref{item:singleton-sm}$ 

  \noindent Follows immediately from Lemma~\ref{lmm:affine},
  due  to  John's  uniqueness  theorem  for  minimum--volume
  enclosing ellipsoids~\cite{John2014}.
\end{proof}

Notice that for (hyper)-ellipsoidal $\bb{S}$, for any $n \in
\N$ and any $\mc{X} \subseteq \R^n$ one has
\begin{align*}
  \ddi(\mc{X}|\bb{S})  =   \left\{  \mvee  \left(   \mc{X}  \right)
  \right\},
\end{align*}
where $\mvee(\mc{X})$ denotes the minimum--volume enclosing
ellipsoid~\cite{todd2016minimum} for $\mc{X}$.

\subsection{Inference of states}\label{s:inference of states}

A  setup  comprising several  boxes  is  given. One  box  is
equipped with $n$  buttons, and the remaining  $m$ boxes are
equipped with one  button and two light bulbs  each. At each
run  of the  experiment,  Alice presses  button  $y$ of  the
former box, and selects box $x$ among the remaining boxes by
pressing its button. She iterates this procedure many times,
recording the frequencies $\{  \mathbf{p}_x \}$ whose $y$-th
element is the  probability of the first bulb of  box $x$ to
light up  given $y$.  This  situation is illustrated  in the
bottom part of Fig.~\ref{fig:experimental-setups}.

The aim  of data-driven  inference is  to infer  the maximal
noncommittal   linear  maps   $R  :   \R^l\rightarrow  \R^n$
consistent with  $\{ \mathbf{p}_x  \}$, that is,  the linear
maps $R$ that  minimize the volume of  $R(\bb{E})$ such that
the  range   $R(\bb{E})$  contains  the   distributions  $\{
\mathbf{p}_x\}\subseteq\R^n$. We recall that $\bb{E}$ is the
set of all effects of the system of linear dimension $\ell$.
Notice that not  any linear map $R$ corresponds to  a set of
physical states:  in case of  a non physical  inference $R$,
failure is declared and a larger set $\{ \mathbf{p}_x \}$ is
required.    This  definition   of  the   problem  naturally
identifies two steps: i) inferring the (possibly non unique)
minimum-volume  range  $\hat{\mc{R}}$  consistent  with  $\{
\mathbf{p}_x \}$, and ii) finding the linear transformations
$\hat{R}$ whose range is $\hat{\mc{R}}$.

In the following definition we  formalize the first of these
two steps, that  is, the inference process  that consists of
finding  the  minimum-volume  range  $R(\bb{E})$  such  that
$R(\bb{E})  \supseteq \{  \mathbf{p}_x \}$.

\begin{dfn}[Data-driven inference of states]
  \label{def:ddi-states}
  For  any set  $\{\mathbf{p}_x\}\subseteq\R^n $,  we denote
  with $\ddi  ( \{  \mathbf{p}_x \}  \; |  \; \bb{E}  )$ the
  \textit{data-driven inference} map
  \begin{align}
    \label{eq:ddi-states}
    \ddi \left(  \left\{ \mathbf{p}_x  \right\} \;  \big| \;
    \bb{E}  \right)  =   \argmin_{\substack{R(\bb{E})\\\{\mathbf{p}_x,\mathbf{0},\mathbf{1}\}  \subseteq
    R(\bb{E}) \subseteq  \spn
    \{\mathbf{p}_x,\mathbf{1}\}
    \\
    R(e)=\mathbf{1}
    }}
    \vol\left(R(\bb{E})\right),
  \end{align}
  where   the  minimization   is  over   subsets  $R(\bb{E})
  \subseteq    \mathbb{R}^n$    corresponding   to    linear
  transformations of the set of effects $\bb{E}$ that lie on
  the         real         span         generated         by
  $\{\mathbf{p}_x\}\cup\{\mathbf{1}\}$,  and the  two points
  $R(0)=\mathbf{0}=(0,0,\ldots,0)\in\R^n$,
  $R(e)=\mathbf{1}=(1,1,\ldots,1)\in\R^n$ are  the images of
  the  null effect  $0\in\bb{E}$  and  of the  deterministic
  $e\in\bb{E}$  effect  respectively  (this  last  condition
  poses     a     not      trivial     linear     constraint
  $R(e)=\mathbf{1}$).
\end{dfn}

In general the output of the  DDI map is not unique, and the
map returns  a set. We  notice that, if the  available prior
information  does  not identify  a  unique  set $\bb{E}$  of
effects, the set  $\bb{E}$ itself can be  considered as part
of  the optimization  problem  in the  above definition,  by
taking  the minimum  of  Eq.~\eqref{eq:ddi-states} over  any
possible $\bb{E}$.

In the definition of DDI of states, the set $\{ \mathbf{p}_x
\}$   has  been   extended   to  include   the  two   points
$\mathbf{0}=(0,0,\ldots,0)\in\R^n$                       and
$\mathbf{1}=(1,1,\ldots,1)\in\R^n$,  corresponding   to  the
null   $0\in\bb{E}$   and   to  the   deterministic   effect
$e\in\bb{E}$ respectively.   This is  so because  the points
$\mathbf{0}$ and $\mathbf{1}$, which could be uncollected by
Alice,  strongly  characterize  the geometry  of  the  range
$R(\bb{E})$   of    the   linear    map   $R$    (see   also
Fig.~\ref{fig:linearmap}).

There are two main differences between the DDI of states and
that         of         measurements        given         in
Definition~\ref{def:ddi-measurements},   when  regarded   as
optimization problems.   The first difference is  in the set
of  points  where the  linear  function  to be  inferred  is
applied.  Indeed, in DDI of states the set $\bb{E}$ is not a
strictly  affine  subspace  of   $\R^l$,  while  in  DDI  of
measurements the  convex set  $\bb{S}$ is a  strictly affine
subspace  of   $\R^l$  of   dimension  $\ell-1$   (see  also
Fig.~\ref{fig:linearmap}). Moreover, the fixed points of the
linear  map $R$  in the  DDI of  states introduce  a further
linear  constraint $R(e)=\mathbf{p}_e$  to the  optimization
problem. Due to these differences we cannot provide a simple
characterization of the DDI of states as for example the one
in Theorem \ref{thm:inference-sm} for DDI of measurements. A
more accurate geometrical analysis of the DDI map for states
will be the subject of future research.

\section{Range inversion}

In both cases of inference  presented above the device to be
inferred  (either  a  measurement  or a  family  of  states)
induces a linear map
\begin{align}
  D:\,&\R^\ell\rightarrow \R^n,
\end{align}
for some $n\in\bb{N}$.

We denote by $\mathcal{A}\in\R^\ell$  a subset of the domain
of  the map  $D$  (this  corresponds to  the  set of  states
$\mathcal{S}$  or the  set of  effects $\mathcal{E}$  in the
inference protocol) and we denote by $D(\mathcal{A})\in\R^n$
the image of $\mathcal{A}$ via the map $D$ (this corresponds
to  the set  of points  $\{\mathbf{p}_x\}_{x=1}^m$ collected
via the inference protocol):
\begin{align*}
  D (\mc{A})  := \left\{ D\mathbf{a}\;  | \;  \mathbf{a}  \in  \mc{A}\right\}.
\end{align*}
Finally we denote  the set of all  linear transformations of
$\mc{A}$ into $\R^n$ as
\begin{align}
\label{eq:Dn}
  \mc{D}_n (\mc{A}) := \left\{ D(
  \mc{A})  \; | \;  \forall D \in  \M{n}{\ell} \;
  \right\}.
\end{align}

We can  now introduce a  notion of equivalence for  maps $D$
based     on    the     coincidence    of     their    range
$D(\bb{A})$.

\begin{dfn}[Equivalence]\label{def:equivalence}
  For any $n \in \N$ and  any $D \in \M{n}{\ell}$, we denote
  with $[D]$ the equivalence class
  \begin{align*}
    \left[ D \right] := \left\{  D' \in \M{n}{\ell} \; |
    \; D'\left( \bb{A} \right) =  D \left(
    \bb{A} \right) \right\}.
  \end{align*}
\end{dfn}

Notice that  any two  elements of  $[D]$ do  not necessarily
share the same support.

In the Main Text we  referred to any transformation $U$ that
leaves  a  set $\bb{A}$  invariant,  that  is such  that  $U
(\bb{A})  =  \bb{A}$,  as  a  \textit{gauge  symmetry}.  The
following  theorem  shows  that  the  equivalence  class  in
Definition \ref{def:equivalence}  is fully specified  by the
symmetries of $\bb{A}$.
\begin{thm}\label{thm:equivalence}
  For any $n \in \N$ and any $D \in \M{n}{\ell}$ one has
  \begin{align}
    \label{eq:equivalence}
    \left[  D \right]  = \left\{  DU \;  | \;  \forall U
    \textrm{ s.t.  } D^+DU \left( \bb{A} \right) =
    D^+D \left( \bb{A} \right) \right\}.
  \end{align}
\end{thm}

\begin{proof}
  The statement can be rephrased as follows. For any $D' \in
  \M{n}{\ell}$ the following conditions are equivalent:
  \begin{enumerate}
  \item\label{item:symmetry}    There    exists    $U    \in
    \M{n}{\ell}$ such that $D' =  D U$ and $D^+D U(\bb{A}) =
    D^+D(\bb{A})$.
  \item\label{item:equivalence} $D' \left(  \bb{A} \right) =
    D\left(\bb{A}\right)$.
  \end{enumerate}

  Let us prove each implication separately:

  \noindent         $\ref{item:symmetry}         \Rightarrow
  \ref{item:equivalence}$

  \noindent By hypothesis $U$ is  such that $D^+D U ( \bb{A}
  ) = D^+D (\bb{A})$.  By multiplying both sides by $D$ from
  the left one  has $DD^+DU ( \bb{A} ) =  DD^+D ( \bb{A} )$.
  Since $D = DD^+D$ one has $DU  ( \bb{A} ) = D ( \bb{A} )$.
  Hence the thesis.

  \noindent        $\ref{item:equivalence}       \Rightarrow
  \ref{item:symmetry}$

  \noindent Since by hypothesis  $D' (\bb{A}) = D (\bb{A})$,
  one has  $\spn D' (\bb{A})  = \spn D (\bb{A})$.   Since by
  hypothesis  $\spn\bb{A}  =  \R^\ell$,   one  has  $\spn  D
  (\bb{A})  = \rng  D$ and  $\spn  D' (\bb{A})  = \rng  D'$.
  Hence, $\rng  D = \rng  D'$, and thus  $D D^+ =  D' D'^+$.
  Then $D'  = D'  D'^+ D'  = D  D^+ D'$.   By setting  $U :=
  D^+D'$ one has $DU (\bb{A}) = D (\bb{A})$.  By multiplying
  both sides by $D^+$ from the  left one has $D^+DU ( \bb{A}
  ) = D^+D ( \bb{A} )$.  Hence the thesis.
\end{proof}

An immediate corollary of the theorem is:
\begin{cor}
  For any $D$ such tat $D^+D = \openone$, due to
  Eq.~\eqref{eq:equivalence} one has that any $D' \in [D]$ is
  equivalent to $D$ up to a symmetry $U$ of the set $\bb{A}$.
\end{cor}

Notice that  if the map  $D$ is  for example the  linear map
associated to  a measurement, the corollary  states that for
any  informationally  complete  measurement  $M$  the  range
$M(\bb{S})$ identifies  $M$ up to gauge  symmetries. This is
the statement of Theorem \ref{thm:range} in the Main Text of
the paper.

\subsection{Measurement range inversion}

In  case of  $(\ell-1)$--dimensional spherical  $\bb{S}$ one
can also  explicitly derive  all the  linear transformations
that correspond to any given (hyper)-ellipsoidal range. This
is the content of the following proposition.

\begin{prop}
  For $(\ell-1)$--dimensional  unit--spherical $\bb{S}$, for
  any  $n  \in   \N$  and  any  $D   \in  \M{n}{\ell}$,  let
  $\mathbf{t}  := D  \mathbf{u}$,  let $T  =  D (\openone  -
  \mathbf{u}\mathbf{u}^T)$, and let $Q = T T^T$. One has
  \begin{align*}
    D \left( \bb{S} \right)  = \left\{  \mathbf{p} \;
    \Big| \;
    \begin{cases}
      (\openone  -  Q  Q^+)   (\mathbf{p}  -  \mathbf{t})  =
      0\\  (\mathbf{p}  -  \mathbf{t})^T Q^+  (\mathbf{p}  -
      \mathbf{t}) \le 1
    \end{cases}
    \right\}.
  \end{align*}
\end{prop}

\begin{proof}
  Due to Lemma~\ref{lmm:affine},  without loss of generality
  we  take  $\bb{S}$  to  be the  unit--sphere  centered  in
  $\mathbf{u}$ (centering  in $\mathbf{u}$  might apparently
  require  a  translation  of   the  sphere  on  the  affine
  subspace, but a linear  transformation suffices due to the
  Lemma; the inverse linear  transformation can be performed
  on the  effects, so  without restriction one  can consider
  the states centered in $\mathbf{u}$).

 One  has  $\mathbf{a}  \in  \bb{S}$   if  and  only  if  $|
 \mathbf{a}  -  \mathbf{u}  |_2  \le  1$  and  $\mathbf{u}^T
 \mathbf{a} = 1$.   For any $\mathbf{a} \in  \bb{S}$ one has
 $D\left(\mathbf{a}\right)  =  \mathbf{t} +  T  \mathbf{a}$.
 Hence
  \begin{align*}
    D\left(\bb{S}\right)   =  \left\{   \mathbf{p}  =
    \mathbf{t} + T \mathbf{a}  \; \Big| \; \left| \mathbf{a}
    - \mathbf{u} \right|_2 = 1, \; \mathbf{u}^T \mathbf{a} =
    1 \right\}.
  \end{align*}
  Solutions of  $T \mathbf{a} = \mathbf{p}  - \mathbf{t}$ in
  variable $\mathbf{a}$ exist if and only if
  \begin{align}
    T  T^+   (\mathbf{p}  -   \mathbf{t})  =   \mathbf{p}  -
    \mathbf{t}.
  \end{align}
  Since $T^+T$  and $\mathbf{u}\mathbf{u}^T$  are orthogonal
  projectors by construction, solutions are given by
  \begin{align*}
    \mathbf{a} = T^+ \left(  \mathbf{p} - \mathbf{t} \right)
    +  \lambda \mathbf{u}  + \left(  \openone -  T^+ T  -
    \mathbf{u}\mathbf{u}^T \right) \mathbf{v},
  \end{align*}
  for any scalar $\lambda$ and any vector $\mathbf{v}$.

  Condition $\mathbf{u}^T \mathbf{a} = 1$ imposes $\lambda =
  1$.     For   any    vector    $\mathbf{v}$   such    that
  $|\mathbf{a}|_2\le 1$, the same condition is also verified
  for  $\mathbf{v} =  0$.   Since $D(\mathbf{a})$  is
  independent of $\mathbf{v}$, without loss of generality we
  take $\mathbf{v} = 0$.
  
  Therefore one has
  \begin{align*}
    D \left(  \bb{S} \right) = \left\{  \mathbf{p} \;
    \Big| \;
    \begin{cases}
      (\openone  -  T  T^+)   (\mathbf{p}  -  \mathbf{t})  =
      0,\\   (\mathbf{p}   -   \mathbf{t})^T   {T^+}^T   T^+
      (\mathbf{p} - \mathbf{t}) \le 1.
    \end{cases}
    \right\}.
  \end{align*}
  By   the  elementary   properties  of   the  Moore-Penrose
  pseudoinverse, one  immediately has that  $T T^+ =  Q Q^+$
  and ${T^+}^T T^+ = Q^+$.  Hence the statement follows.
\end{proof}

\section{Data-driven reconstruction}

\subsection{Data-driven reconstruction of measurements}

In   the   protocol   of   data-driven   reconstruction   of
measurements, an  experimentalist, say Bob, is  in charge of
building the state-preparator  $\mc{S}$ corresponding to the
box   equipped   with  buttons   in   the   upper  part   of
Fig.~\ref{fig:experimental-setups-2}.  His aim  is to enable
Alice to correctly infer  measurement $M$, corresponding to
the box with  light bulbs, up to the  equivalence of Theorem
\ref{thm:equivalence}. In this case, we say that $\mc{S}$ is
observationally complete for $M$.

\begin{figure}[t!]
  \includegraphics[width=0.9\columnwidth]{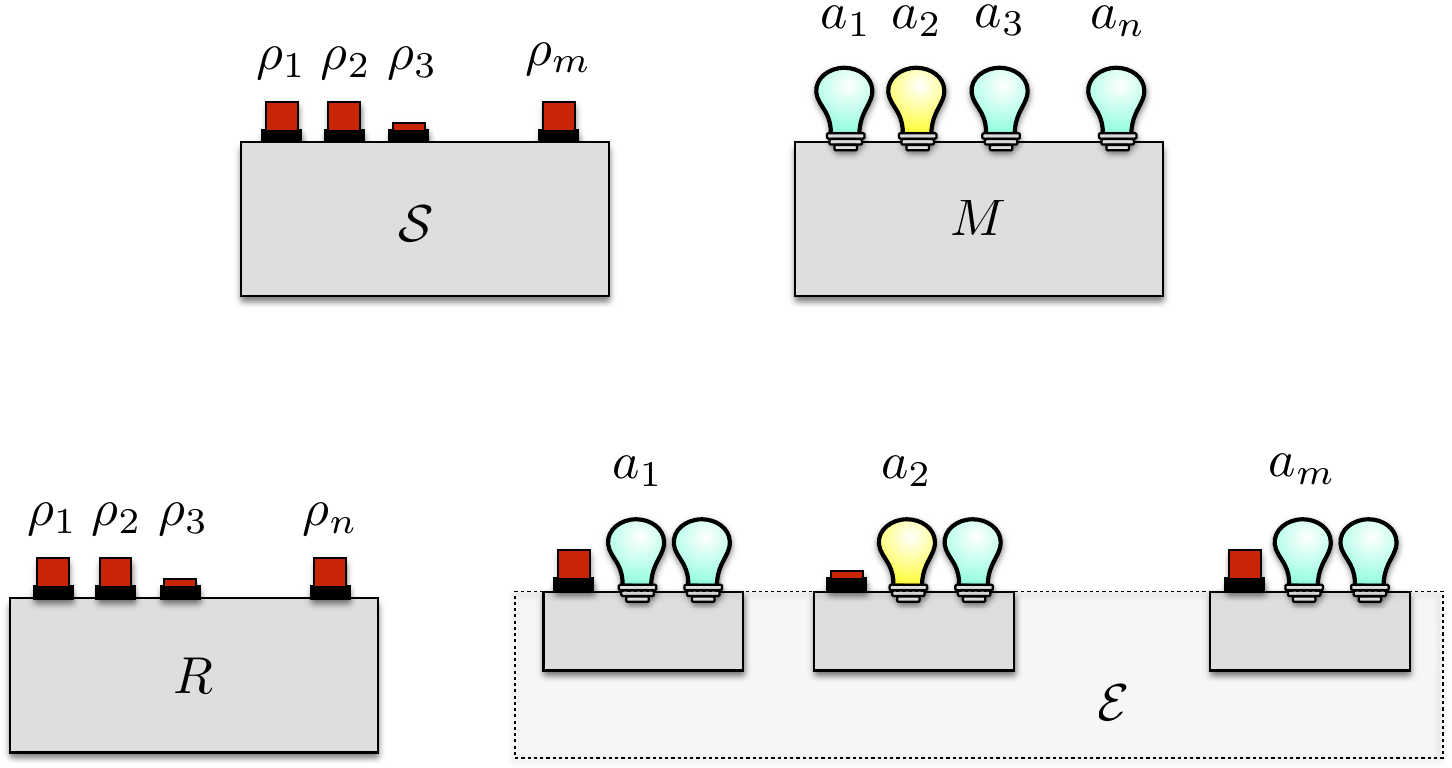}
  \caption{{\bf    Top}:   Setup    for   the    data-driven
    reconstruction  of  measurements.   On the  right  is  a
    measurement $\{a_y\}_{y=1}^{n}$, that  can be thought as
    a box $M$  equipped with $n$ bulbs, one  for each effect
    $a_y\in \mathbb{E}$, $y=1,2,\ldots, n$, corresponding to
    the possible outcomes of the measurement. On the left is
    a state preparator that can be thought as a box with $m$
    possible buttons,  each button corresponding to  a state
    $\rho_x\in\bb{S}$, $x=1,\ldots,m$.   At each run  of the
    experiment a  button is pressed  and the outcome  of the
    measurements  is recorded.   We  iterate this  procedure
    many times  recording the  frequencies $p_x^y$  that are
    our estimates of the probabilities $a_y(\rho_x)$.  Since
    each state  $\rho_x\in\mathbb{S}$ can  be regarded  as a
    vector  in  $\R^\ell$,  the   whole  experiment  can  be
    represented  by a  linear map  $M:\R^l\rightarrow \R^n$,
    with   $\rho  _x   \mapsto   \mathbf{p}_x$,  $p_x^y   :=
    a_y(\rho_x)$.      Any     state      in     the     set
    $\mathcal{S}=\{\rho_x\}_{x=1}^m$  is  associated with  a
    point  $\mathbf{p}_x\in\R^n$.  {\bf  Bottom}: Setup  for
    the  data-driven reconstruction  of families  of states.
    On the left  is a state preparator $R$  that can prepare
    one  out of  a finite  set  $\{ \rho_y  \}_{y =1}^n$  of
    states. We can think of an $n$-buttons state preparator:
    whenever we press  the button $y$ the  state $\rho_y$ is
    prepared.  On  the right is  a measurement box  with $m$
    possible    buttons:   whenever    button   $x$    (with
    $x=1,\ldots,m$) is pressed  the two outcomes measurement
    described by the pair of  effects $\{a_x, \bar a_x\}$ is
    applied to the prepared state. Give a state $\rho_y$ the
    outcome   $a_x$  will   be  obtained   with  probability
    $a_x(\rho_y)$ (clearly, $a_x(\rho_y)  + \bar a_x(\rho_y)
    =1 $).  We iterate this procedure for all the $n$ states
    $\rho_y$ and  the $m$ measurements $\{a_x,  \bar a_x\}$,
    thus recording  the frequencies $ p^y_x$,  which are our
    estimate  of the  probabilities  $ a_x(\rho_y)$.   Since
    each   effect  $a_x$   corresponds   to   a  vector   in
    $\mathbb{R}^\ell$,   the   whole   experiment   can   be
    represented  by  a   linear  map  $R:\mathbb{R}^{l}  \to
    \mathbb{R}^{n}$,   with   $a_x  \mapsto   \mathbf{p}_x$,
    ${p}^y_x  :=  a_x(\rho_y)$.   Any   effect  in  the  set
    $\mathcal{E}=\{a_x\}_{x=1}^m$ is associated with a point
    $\mathbf{p}_x\in\R^n$.}
  \label{fig:experimental-setups-2}
\end{figure}

\begin{dfn}[Observational complete set of states]
  \label{def:completeness-sm}
  Let $\{\mathbb{S},  \mathbb{E}\}$ be a physical  system of
  linear  dimension $\ell$  and $M  \in \M{n}{\ell}  $ be  a
  measurement.   A  set  of  states  $\mathcal{S}  \subseteq
  \mathbb{S}$ is \textit{observationally  complete for $M$} if
  and only if
  \begin{align}
    \label{eq:completeness-sm}
    \ddi( M  \left(  \mc{S}  \right) | \mathbb{S})= 
    M \left( \bb{S} \right)
  \end{align}
  where $\ddi$ is the data driven inference of measurements.
\end{dfn}

The following result shows  that the notion of observational
completeness depends only on the support of $M$.

\begin{thm}\label{thm:completenessgeneral-sm}
  Let $\{\mathbb{S},  \mathbb{E}\}$ be a physical  system of
  linear  dimension $\ell$  and  let $\mathcal{S}  \subseteq
  \mathbb{S}$ be  a set of  states.  Let $\mathcal{V}$  be a
  linear subspace of $\mathbb{R}^\ell$  and let $\Pi$ denote
  the  projector on  $\mathcal{V}$.   Then $\mathcal{S}$  is
  observationally complete  for $\Pi$ if  and only if  it is
  observationally complete for any  $M \in \M{n}{\ell}$ such
  that $\supp M = \mathcal{V}$, i.e.
  \begin{align}
    \begin{aligned}
      & \ddi(  \Pi \left( \mc{S} \right)  | \mathbb{S})= \Pi
      \left( \bb{S} \right) \iff \\  & \ddi( M \left( \mc{S}
      \right)  | \mathbb{S})=  M  \left(  \bb{S} \right)  \;
      \forall M \mbox{ s.t. } \supp M = \mathcal{V}.
    \end{aligned}
  \end{align}
\end{thm}

\begin{proof}
  We only need  to prove the $\implies$  direction since the
  opposite one is trivially true.   Let us then suppose that
  $  \ddi(  \Pi \left(  \mc{S}  \right)  | \mathbb{S})=  \Pi
  \left( \bb{S} \right)$  and let us fix an  arbitrary $M \in
  \M{n}{\ell}$ such  that $\supp M =  \mathcal{V}$.  Then we
  have     $M    \Pi     =    M$.      By    using     lemma
  \ref{lmm:commutativity-sm} we have
  \begin{align*}
    \begin{aligned}
      &\ddi( M \left( \mc{S}  \right) | \mathbb{S})= \ddi( M
      \Pi \left(  \mc{S} \right) | \mathbb{S})=\\  & M \ddi(
      \Pi \left( \mc{S} \right)  | \mathbb{S})= M \Pi \left(
      \bb{S} \right) = M \left( \bb{S} \right)
    \end{aligned}
  \end{align*}
  and the thesis is proved.
\end{proof}

An immediate corollary of this theorem is:

\begin{cor}\label{thm:completeness-sm}
  A  set of  states  $\mathcal{S}  \subseteq \mathbb{S}$  is
  observationally complete for  any informationally complete
  measurement if  and only  if $\ddi( \mc{S}  | \mathbb{S})=
  \bb{S}$.
\end{cor}

This is  the statement of  Theorem~\ref{thm:completeness} in
the Main Text.

If the  set of states $\mathbb{S}$  is (hyper)-spherical, it
is possible to give an explicit characterization of the sets
of states with minimum  cardinality that are observationally
complete for the informationally complete measurements.

\begin{prop}
  Let $\{\mathbb{S},  \mathbb{E}\}$ be a physical  system of
  linear  dimension  $\ell$  and   let  $\mathbb{S}$  be  an
  $(\ell-1)$-dimensional (hyper)-sphere.  Then the following
  conditions are equivalent:
  \begin{enumerate}
  \item\label{item:simplex} $\mc{S}$ is a regular
    $(\ell-1)$--simplex inscribed in $\bb{S}$.
  \item\label{item:minimal} $\ddi( \mc{S} | \mathbb{S})=
    \bb{S} $ and $\mc{S}$ has minimal cardinality.
  \end{enumerate}
\end{prop}

\begin{proof}
  Let us prove each implication separately:
  
  \noindent         \ref{item:simplex}         $\Rightarrow$
  \ref{item:minimal}
  
  \noindent Let $\mc{A}$  be the regular $(\ell-1)$--simplex
  inscribed        in         $\bb{A}$.         It        is
  known~\cite{vandev1992minimal}   that   $\mvee(\mc{A})   =
  \bb{A}$.
  
  \noindent         \ref{item:minimal}         $\Rightarrow$
  \ref{item:simplex}
  
  \noindent Since  $\mathbb{S}$ is  (hyper)-spherical $\ddi(
  \mc{S}  | \mathbb{S})  = \mvee(  \mc{S}) =  \bb{S}$.  Then
  $\mvee(  \mc{S})$  is  the smallest  (hyper)-sphere  which
  contains  $\mc{S}$.  Clearly  the cardinality  of $\mc{S}$
  must be  greater then $\ell$.   On the other hand,  as the
  proof of the previous item  shows, the regular simplex has
  cardinality  $\ell$   and  is   observationally  complete.
  Therefore $\mc{S}$  must be a  simplex.  We now  show that
  $\mc{S}$ is  regular.  Let us denote  with $\conv(\mc{S})$
  the  convex hull  of  $\mc{S}$ and  let  $r_{max}$ be  the
  radius of the largest sphere inscribed in $\conv(\mc{S})$.
  Since $ \mvee( \mc{S}) \subseteq (\ell -1) \conv(\mc{S}) $
  \cite{boyd2004convex} we have $ (\ell -1) r_{max} \geq R $
  where  $R$ is  the radius  of $\mvee(  \mc{S}) $.   On the
  other hand,  we have  $ (\ell  -1) r_{max}  \leq R  $ from
  Euler inequality ~\cite{Vince2008}.  Therefore $ (\ell -1)
  r_{max} =  R$ which holds  if and  only if the  simplex is
  regular.
\end{proof}

Let us consider the case when  $\ell = 3$ and hence $\bb{S}$
is a  circle.  In  this case,  any regular  polygon $\mc{S}$
with    $n$    vertices    inscribed    in    $\bb{A}$    is
$\mc{U}$--symmetric,   where  $\mc{U}$   is  an   orthogonal
representation of the dihedral group.   Since for $n \ge 3$,
the only  $\mc{U}$--symmetric ellipse is the  circle, due to
Theorem~\ref{thm:inference-sm}                           and
Corollary~\ref{thm:completeness-sm} any such  an $\mc{S}$ is
observationally  complete for  any informationally  complete
measurement.

Let  us now  consider the  case when  $\ell =  4$ and  hence
$\bb{S}$ is a sphere.  In this case, any Platonic solid with
$n$  vertices inscribed  in $\bb{S}$  is $\mc{U}$-symmetric,
where  $\mc{U}$  is  an  orthogonal  representation  of  the
tetrahedral (for tetrahedra),  octahedral (for octahedra and
cubes),  or  icosahedral  (for  icosahedra  or  dodecahedra)
group.  Since the only  $\mc{U}$--symmetric ellipsoid is the
sphere,     due    to     Theorem~\ref{thm:inference}    and
Corollary~\ref{thm:completeness-sm} any such  an $\mc{S}$ is
observationally  complete for  any informationally  complete
measurement.

\subsection{Data-driven reconstruction of states}

In the  protocol of data-driven reconstruction  of family of
states,   Bob   is  in   charge   of   building  the   tests
(binary-outcome measurements) $\mc{E}$  corresponding to the
boxes equipped with  one button and two light  bulbs each in
the lower part of Fig.~\ref{fig:experimental-setups-2}.  His
aim  is to  enable Alice  to correctly  infer the  family of
states $R$, corresponding to the box with $n$ buttons, up to
the equivalence  of Theorem \ref{thm:equivalence}.   In this
case, we  say that $\mc{E}$ is  observationally complete for
$R$.

\begin{dfn}[Observational complete set of effects]
  \label{def:completenesseffect-sm}
  Let $\{\mathbb{S},  \mathbb{E}\}$ be a physical  system of
  linear  dimension $\ell$  and $R  \in \M{n}{\ell}  $ be  a
  family of states.  A set of effects $\mathcal{E} \subseteq
  \mathbb{E}$ is \textit{observationally  complete for $R$} if
  and only if
  \begin{align}
    \label{eq:completeness-sm}
    \ddi( R  \left(  \mc{E}  \right) | \mathbb{E})= 
    R \left( \bb{E} \right)
  \end{align}
  where $\ddi$ is the data driven inference of states.
\end{dfn}

Clearly,        the        analogous       of        theorem
\ref{thm:completenessgeneral-sm} holds

\begin{thm}
  Let $\{\mathbb{R},  \mathbb{E}\}$ be a physical  system of
  linear  dimension $\ell$  and  let $\mathcal{E}  \subseteq
  \mathbb{E}$ be a  set of effects.  Let  $\mathcal{V}$ be a
  linear subspace of $\mathbb{R}^\ell$  and let $\Pi$ denote
  the  projector on  $\mathcal{V}$.   Then $\mathcal{E}$  is
  observationally complete  for $\Pi$ if  and only of  it is
  observationally complete for any  $R \in \M{n}{\ell}$ such
  that $\supp R = \mathcal{V}$, i.e.
  \begin{align}
    \begin{aligned}
      & \ddi(  \Pi \left( \mc{E} \right)  | \mathbb{E})= \Pi
      \left( \bb{E} \right) \iff \\  & \ddi( R \left( \mc{E}
      \right)  |  \mathbb{E})=  R\left(  \bb{E}  \right)  \;
      \forall R \mbox{ s.t. } \supp R = \mathcal{V}.
    \end{aligned}
  \end{align}
\end{thm}

\begin{proof}
  The proof of this result is completely analogous to the
  proof of Theorem~\ref{thm:completenessgeneral-sm}.
\end{proof}

\section{Technical Lemmas}

\subsection{Affine transformations of a strictly affine set}

For  any  $n \in  \N$,  any  $D  \in \M{n}{\ell}$,  and  any
$\mathbf{d} \in \R^n$,  let $\mc{F}_{D, \mathbf{d}}: \R^\ell
\to  \R^n$  denote   the  affine  map  such   that  for  any
$\mathbf{a} \in \R^\ell$ one has
\begin{align*}
  \mc{F}_{D,\mathbf{d}}(\mathbf{a})   :=   D  \mathbf{a}   +
  \mathbf{d}.
\end{align*}
For  any set  $\mc{A} \subseteq  \R^\ell$ we  adopt the  set
builder notation
\begin{align*}
  \mc{F}_{D,\mathbf{d}}  \left(  \mc{A} \right)  :=  \left\{
  \mc{F}_{D,\mathbf{d}} \left(  \mathbf{a} \right)  \; |
  \; \mathbf{a} \in \mc{A} \right\}.
\end{align*}
For   any  $n   \in  \bb{N}$   and  any   $\mc{A}  \subseteq
\bb{R}^\ell$, let  $\mc{F}_n(\mc{A})$ denote the set  of all
affine transformations of $\mc{A}$ into $\R^n$, that is
\begin{align*}
  \mc{F}_n   \left(\mc{A}\right)   :=   \left\{   \mc{F}_{D,
    \mathbf{d}} \left( \mc{A} \right)  \; \Big| \; \forall D
  \in  \M{n}{\ell}, \;  \forall \mathbf{d}  \in \bb{R}^n  \;
  \right\}.
\end{align*}

We say that  an affine subspace is {\em  strictly} affine if
and only if it is not a  linear subspace. For any $n \in \N$
and  any  $\mc{X} \in  \R^n$,  let  $\aff\mc{X}$ denote  the
affine hull  of $\mc{X}$.  Let $\mathbf{u}  \in \aff \bb{A}$
denote the vector orthogonal to $\aff \bb{A}$.  Without loss
of generality we take $\mathbf{u}$ such that $|\mathbf{u}|_2
= 1$.

\begin{lmm}
  \label{lmm:affine}
  For any $n  \in \N$ and any $\bb{A}$ subset  of a strictly
  affine subspace of $\bb{R}^\ell$ one has
  \begin{align*}
    \mathcal{D}_n  \left( \bb{A}  \right) =  \mc{F}_n \left(
    \bb{A} \right).
  \end{align*}
(see Eq.~\eqref{eq:Dn} for the definition of $\mathcal{D}_n$.)
\end{lmm}

\begin{proof}
  Of course  $\mathcal{D}_n ( \bb{A} )  \subseteq \mc{F}_n (
  \bb{A}  )$,  so   we  only  need  to   prove  the  inverse
  inclusion. For any $D \in \M{n}{\ell}$ and $\mathbf{d} \in
  \R^n$, one has that $D' := D + \mathbf{d} \mathbf{u}^T$ is
  such that  $\mc{F}_{D, \mathbf{d}}(a) = D'(a)$  for any $a
  \in \bb{A}$. Hence the thesis follows.
\end{proof}

\subsection{Commutativity of the DDI map}

\begin{lmm}
  \label{lmm:commutativity-sm}
Let  $D \in  \M{n}{\ell}$, and  
  $\mc{A}  \subseteq \R^\ell$  such  that $\mc{A}  \subseteq
  \supp D$ one has
  \begin{align}
    \label{eq:commutativity-sm}
    D ( \ddi  (\mc{A} | \mathbb{A} ) )=
    \ddi (D
      (\mc{A})| \mathbb{A}).
  \end{align}
\end{lmm}
\begin{proof}
  By    definition,    the    l.h.s.    and    r.h.s.     of
  Eq.~\eqref{eq:commutativity-sm}     are      given     by,
  respectively,
  \begin{align*}
    D ( \ddi  (\mc{A} | \mathbb{A} ) )   &  :=
                                 D  (   \argmin_{\mc{X}  \in \mathfrak{X}}  \vol ( \mc{X} ) ),\\
  \ddi (D      (\mc{A})| \mathbb{A})   &  :=
                                \argmin_{\mc{Y}  \in\mathfrak{Y}} \vol ( \mc{Y} ),
  \end{align*}
  with
  \begin{align*}
    \mathfrak{X}       & :=     \{       \mc{X}      \in
    \mc{D}_\ell\left(\bb{A}\right)   \;   |  \;   \mc{A}
    \subseteq \mc{X} \subseteq \aff \mc{A} \},\\
    \mathfrak{Y}       & :=      \{       \mc{Y}      \in
    \mc{D}_n\left(\bb{A}\right) \; |  \; D \left(
    \mc{A} \right) \subseteq  \mc{Y} \subseteq \aff D
    \left( \mc{A} \right) \}.
  \end{align*}
  The map  $D$   is  bijective  from   $\mathfrak{X}$  to
  $\mathfrak{Y}$. This can be easily seen as follows.  Since
  $\mc{A}   \subseteq    \supp   D$,   by    definition   of
  $\mathfrak{X}$ for  any $\mc{X} \in \mathfrak{X}$  one has
  $\mc{X}  \subseteq  \supp  D$.   Also,  by  definition  of
  $\mathfrak{Y}$ for  any $\mc{Y} \in \mathfrak{Y}$  one has
  $\mc{Y} \subseteq \supp D^+$.
    Moreover, $D$ preserves the  ordering induced by function
  $\vol$, that is:
  \begin{align*}
    \vol  \left(  \mc{X}  \right) =  \lambda_D  \vol  \left(
    \mc{L}_D \left( \mc{X} \right) \right) \quad \forall \mc{X} \in \mathfrak{X},
  \end{align*}
  for some $\lambda_D > 0$ that only depends on $D$.
  Hence the statement follows.
\end{proof}

\bibliography{bibliography}

\end{document}